\documentclass[twoside,11pt]{article}
\usepackage{amsmath}
\usepackage{amsthm}
\usepackage{amssymb}
\usepackage{amscd}
\usepackage{graphicx}
\usepackage{epsfig}
\usepackage{bm}

\usepackage{latexsym}
\usepackage{amsfonts}

\input xy
\xyoption{all}

=1in \textwidth 15.6cm \topmargin 15mm\textheight23cm \advance\topmargin by -\headheight\advance\topmargin by
-\headsep

\numberwithin{equation}{section} \allowdisplaybreaks

\newtheorem{theo}{Theorem}[section]
\newtheorem{prop}{Proposition}[section]

\theoremstyle{definition}
\newtheorem{deff}{Definition}[section]

\newtheorem{example}{Example}[section]
\newtheorem{rem}{Remark}[section]

\begin{document}
\font\black=cmbx10 \font\sblack=cmbx7 \font\ssblack=cmbx5 \font\blackital=cmmib10  \skewchar\blackital='177
\font\sblackital=cmmib7 \skewchar\sblackital='177 \font\ssblackital=cmmib5 \skewchar\ssblackital='177
\font\sanss=cmss10 \font\ssanss=cmss8 
\font\sssanss=cmss8 scaled 600 \font\blackboard=msbm10 \font\sblackboard=msbm7 \font\ssblackboard=msbm5
\font\caligr=eusm10 \font\scaligr=eusm7 \font\sscaligr=eusm5 \font\blackcal=eusb10 \font\fraktur=eufm10
\font\sfraktur=eufm7 \font\ssfraktur=eufm5 \font\blackfrak=eufb10

\font\bsymb=cmsy10 scaled\magstep2
\def\all#1{\setbox0=\hbox{\lower1.5pt\hbox{\bsymb
       \char"38}}\setbox1=\hbox{$_{#1}$} \box0\lower2pt\box1\;}
\def\exi#1{\setbox0=\hbox{\lower1.5pt\hbox{\bsymb \char"39}}
       \setbox1=\hbox{$_{#1}$} \box0\lower2pt\box1\;}

\def\mi#1{{\fam1\relax#1}}
\def\tx#1{{\fam0\relax#1}}

\newfam\bifam
\textfont\bifam=\blackital \scriptfont\bifam=\sblackital \scriptscriptfont\bifam=\ssblackital
\def\bi#1{{\fam\bifam\relax#1}}

\newfam\blfam
\textfont\blfam=\black \scriptfont\blfam=\sblack \scriptscriptfont\blfam=\ssblack
\def\rbl#1{{\fam\blfam\relax#1}}

\newfam\bbfam
\textfont\bbfam=\blackboard \scriptfont\bbfam=\sblackboard \scriptscriptfont\bbfam=\ssblackboard
\def\bb#1{{\fam\bbfam\relax#1}}

\newfam\ssfam
\textfont\ssfam=\sanss \scriptfont\ssfam=\ssanss \scriptscriptfont\ssfam=\sssanss
\def\sss#1{{\fam\ssfam\relax#1}}

\newfam\clfam
\textfont\clfam=\caligr \scriptfont\clfam=\scaligr \scriptscriptfont\clfam=\sscaligr
\def\cl#1{{\fam\clfam\relax#1}}

\newfam\frfam
\textfont\frfam=\fraktur \scriptfont\frfam=\sfraktur \scriptscriptfont\frfam=\ssfraktur
\def\fr#1{{\fam\frfam\relax#1}}

\def\cb#1{\hbox{$\fam\gpfam\relax#1\textfont\gpfam=\blackcal$}}

\def\hpb#1{\setbox0=\hbox{${#1}$}
    \copy0 \kern-\wd0 \kern.2pt \box0}
\def\vpb#1{\setbox0=\hbox{${#1}$}
    \copy0 \kern-\wd0 \raise.08pt \box0}

\def\pmb#1{\setbox0\hbox{${#1}$} \copy0 \kern-\wd0 \kern.2pt \box0}
\def\pmbb#1{\setbox0\hbox{${#1}$} \copy0 \kern-\wd0
      \kern.2pt \copy0 \kern-\wd0 \kern.2pt \box0}
\def\pmbbb#1{\setbox0\hbox{${#1}$} \copy0 \kern-\wd0
      \kern.2pt \copy0 \kern-\wd0 \kern.2pt
    \copy0 \kern-\wd0 \kern.2pt \box0}
\def\pmxb#1{\setbox0\hbox{${#1}$} \copy0 \kern-\wd0
      \kern.2pt \copy0 \kern-\wd0 \kern.2pt
      \copy0 \kern-\wd0 \kern.2pt \copy0 \kern-\wd0 \kern.2pt \box0}
\def\pmxbb#1{\setbox0\hbox{${#1}$} \copy0 \kern-\wd0 \kern.2pt
      \copy0 \kern-\wd0 \kern.2pt
      \copy0 \kern-\wd0 \kern.2pt \copy0 \kern-\wd0 \kern.2pt
      \copy0 \kern-\wd0 \kern.2pt \box0}

\def\cdotss{\mathinner{\cdotp\cdotp\cdotp\cdotp\cdotp\cdotp\cdotp
        \cdotp\cdotp\cdotp\cdotp\cdotp\cdotp\cdotp\cdotp\cdotp\cdotp
        \cdotp\cdotp\cdotp\cdotp\cdotp\cdotp\cdotp\cdotp\cdotp\cdotp
        \cdotp\cdotp\cdotp\cdotp\cdotp\cdotp\cdotp\cdotp\cdotp\cdotp}}

\font\frak=eufm10 scaled\magstep1 \font\fak=eufm10 scaled\magstep2 \font\fk=eufm10 scaled\magstep3
\font\scriptfrak=eufm10 \font\tenfrak=eufm10


\mathchardef\za="710B  
\mathchardef\zb="710C  
\mathchardef\zg="710D  
\mathchardef\zd="710E  
\mathchardef\zve="710F 
\mathchardef\zz="7110  
\mathchardef\zh="7111  
\mathchardef\zvy="7112 
\mathchardef\zi="7113  
\mathchardef\zk="7114  
\mathchardef\zl="7115  
\mathchardef\zm="7116  
\mathchardef\zn="7117  
\mathchardef\zx="7118  
\mathchardef\zp="7119  
\mathchardef\zr="711A  
\mathchardef\zs="711B  
\mathchardef\zt="711C  
\mathchardef\zu="711D  
\mathchardef\zvf="711E 
\mathchardef\zq="711F  
\mathchardef\zc="7120  
\mathchardef\zw="7121  
\mathchardef\ze="7122  
\mathchardef\zy="7123  
\mathchardef\zf="7124  
\mathchardef\zvr="7125 
\mathchardef\zvs="7126 
\mathchardef\zf="7127  
\mathchardef\zG="7000  
\mathchardef\zD="7001  
\mathchardef\zY="7002  
\mathchardef\zL="7003  
\mathchardef\zX="7004  
\mathchardef\zP="7005  
\mathchardef\zS="7006  
\mathchardef\zU="7007  
\mathchardef\zF="7008  
\mathchardef\zW="700A  

\newcommand{\be}{\begin{equation}}
\newcommand{\ee}{\end{equation}}
\newcommand{\ra}{\rightarrow}
\newcommand{\lra}{\longrightarrow}
\newcommand{\bea}{\begin{eqnarray}}
\newcommand{\eea}{\end{eqnarray}}
\newcommand{\beas}{\begin{eqnarray*}}
\newcommand{\eeas}{\end{eqnarray*}}
\def\*{{\textstyle *}}
\newcommand{\R}{{\mathbb R}}
\newcommand{\T}{{\mathbb T}}
\newcommand{\C}{{\mathbb C}}
\newcommand{\unit}{{\mathbf 1}}
\newcommand{\SL}{SL(2,\C)}
\newcommand{\Sl}{sl(2,\C)}
\newcommand{\SU}{SU(2)}
\newcommand{\su}{su(2)}
\def\ssT{\sss T}
\newcommand{\G}{{\goth g}}
\newcommand{\D}{{\rm d}}
\newcommand{\Df}{{\rm d}^\zF}
\newcommand{\de}{\,{\stackrel{\rm def}{=}}\,}
\newcommand{\we}{\wedge}
\newcommand{\nn}{\nonumber}
\newcommand{\ot}{\otimes}
\newcommand{\s}{{\textstyle *}}
\newcommand{\ts}{T^\s}
\newcommand{\oX}{\stackrel{o}{X}}
\newcommand{\oD}{\stackrel{o}{D}}
\newcommand{\obD}{\stackrel{o}{\bD}}
\newcommand{\pa}{\partial}
\newcommand{\ti}{\times}
\newcommand{\A}{{\cal A}}
\newcommand{\Li}{{\cal L}}
\newcommand{\ka}{\mathbb{K}}
\newcommand{\find}{\mid}
\newcommand{\ad}{{\rm ad}}
\newcommand{\rS}{]^{SN}}
\newcommand{\rb}{\}_P}
\newcommand{\p}{{\sf P}}
\newcommand{\h}{{\sf H}}
\newcommand{\X}{{\cal X}}
\newcommand{\I}{\,{\rm i}\,}
\newcommand{\rB}{]_P}
\newcommand{\Ll}{{\pounds}}
\def\lna{\lbrack\! \lbrack}
\def\rna{\rbrack\! \rbrack}
\def\rnaf{\rbrack\! \rbrack_\zF}
\def\rnah{\rbrack\! \rbrack\,\hat{}}
\def\lbo{{\lbrack\!\!\lbrack}}
\def\rbo{{\rbrack\!\!\rbrack}}
\def\lan{\langle}
\def\ran{\rangle}
\def\zT{{\cal T}}
\def\tU{\tilde U}
\def\ati{{\stackrel{a}{\times}}}
\def\sti{{\stackrel{sv}{\times}}}
\def\aot{{\stackrel{a}{\ot}}}
\def\sati{{\stackrel{sa}{\times}}}
\def\saop{{\stackrel{sa}{\op}}}
\def\bwa{{\stackrel{a}{\bigwedge}}}
\def\svop{{\stackrel{sv}{\oplus}}}
\def\saot{{\stackrel{sa}{\otimes}}}
\def\cti{{\stackrel{cv}{\times}}}
\def\cop{{\stackrel{cv}{\oplus}}}
\def\dra{{\stackrel{\xd}{\ra}}}
\def\bdra{{\stackrel{\bd}{\ra}}}
\def\bAff{\mathbf{Aff}}
\def\Aff{\sss{Aff}}
\def\bHom{\mathbf{Hom}}
\def\Hom{\sss{Hom}}
\def\bt{{\boxtimes}}
\def\sot{{\stackrel{sa}{\ot}}}
\def\bp{{\boxplus}}
\def\op{\oplus}
\def\bwak{{\stackrel{a}{\bigwedge}\!{}^k}}
\def\aop{{\stackrel{a}{\oplus}}}
\def\ix{\operatorname{i}}
\def\V{{\cal V}}
\def\cD{{\cal D}}
\def\cC{{\cal C}}
\def\cE{{\cal E}}
\def\cL{{\cal L}}
\def\cN{{\cal N}}
\def\cR{{\cal R}}
\def\cJ{{\cal J}}
\def\cT{{\cal T}}
\def\cH{{\cal H}}
\def\bA{\mathbf{A}}
\def\bI{\mathbf{I}}
\def\wh{\widehat}
\def\wt{\widetilde}
\def\ol{\overline}
\def\ul{\underline}
\def\Sec{\operatorname{Sec}}
\def\Lin{\sss{Lin}}
\def\ader{\sss{ADer}}
\def\ado{\sss{ADO^1}}
\def\adoo{\sss{ADO^0}}
\def\AS{\sss{AS}}
\def\bAS{\sss{AS}}
\def\bLS{\sss{LS}}
\def\bAP{\sss{AV}}
\def\bLP{\sss{LP}}
\def\AP{\sss{AP}}
\def\LP{\sss{LP}}
\def\LS{\sss{LS}}
\def\Z{\mathbf{Z}}
\def\oZ{\overline{\bZ}}
\def\oA{\overline{\bA}}
\def\cim{{C^\infty(M)}}
\def\de{{\cal D}^1}
\def\la{\langle}
\def\ran{\rangle}
\def\by{{\bi y}}
\def\bs{{\bi s}}
\def\bc{{\bi c}}
\def\bd{{\bi d}}
\def\bh{{\bi h}}
\def\bD{{\bi D}}
\def\bY{{\bi Y}}
\def\bX{{\bi X}}
\def\bL{{\bi L}}
\def\bV{{\bi V}}
\def\bW{{\bi W}}
\def\bS{{\bi S}}
\def\bT{{\bi T}}
\def\bC{{\bi C}}
\def\bE{{\bi E}}
\def\bF{{\bi F}}
\def\bP{{\bi P}}
\def\bp{{\bi p}}
\def\bz{{\bi z}}
\def\bZ{{\bi Z}}
\def\bq{{\bi q}}
\def\bQ{{\bi Q}}
\def\bx{{\bi x}}

\def\sA{{\sss A}}
\def\sC{{\sss C}}
\def\sD{{\sss D}}
\def\sG{{\sss G}}
\def\sH{{\sss H}}
\def\sI{{\sss I}}
\def\sJ{{\sss J}}
\def\sK{{\sss K}}
\def\sL{{\sss L}}
\def\sO{{\sss O}}
\def\sP{{\sss P}}
\def\sPh{{\sss P\sss h}}
\def\sT{{\sss T}}
\def\sV{{\sss V}}
\def\sR{{\sss R}}
\def\sS{{\sss S}}
\def\sE{{\sss E}}
\def\sF{{\sss F}}
\def\st{{\sss t}}
\def\sg{{\sss g}}
\def\sx{{\sss x}}
\def\sv{{\sss v}}
\def\sw{{\sss w}}
\def\sQ{{\sss Q}}
\def\sj{{\sss j}}
\def\sq{{\sss q}}
\def\xa{\tx{a}}
\def\xc{\tx{c}}
\def\xd{\tx{d}}
\def\xi{\tx{i}}
\def\xD{\tx{D}}
\def\xV{\tx{V}}
\def\xF{\tx{F}}
\def\dt{\xd_{\sss T}}
\def\dte{\dt^\ze}
\def\dtwte{\dt^{\wt{\ze}}}
\def\vt{\textsf{v}_{\sss T}}
\def\vta{\operatorname{v}_\zt}
\def\vtb{\operatorname{v}_\zp}
\def\cM{{\cal M}}
\def\cN{{\cal N}}
\def\cD{{\cal D}}
\def\ug{{\ul{\zg}}}
\def\rel{{-\!\!\!-\!\!\rhd}}
\newdir{|>}{%
!/4.5pt/@{|}*:(1,-.2)@^{>}*:(1,+.2)@_{>}}
\def\rk{\operatorname{rank}}
\def\dim{\operatorname{dim}}
\def\Graph{\operatorname{graph}}
\def\Dom{\operatorname{Dom}}
\def\Rg{\operatorname{Rg}}
\def\Ph{\sss{Ph}}
\def\Ve{\sss{Vel}}
\def\el{{\cE}}
\def\pr{\operatorname{pr}}
\def\suit{\sss{Suit}}
\def\Der{\sss{Der}}
\def\Suit{\operatorname{Suit}}
\def\bzg{{\bm{\wh{\gamma}}}}

\setcounter{page}{1} \thispagestyle{empty}


\bigskip

\bigskip

\title{Dirac Algebroids\\ in Lagrangian and Hamiltonian Mechanics\thanks{Research
supported by the Polish Ministry of Science and Higher Education under the grant N N201 365636.} }

        \author{
        Katarzyna  Grabowska$^1$, Janusz Grabowski$^2$\\
        \\
         $^1$ {\it Faculty of Physics}\\
                {\it University of Warsaw} \\ \\
         $^2$ {\it Institute of Mathematics}\\
                {\it Polish Academy of Sciences}
                }
\date{}
\maketitle
\begin{abstract}
We present a unified approach to constrained implicit Lagrangian and Hamiltonian systems based on the
introduced concept of {\it Dirac algebroid}. The latter is a certain almost Dirac structure associated with
the Courant algebroid $\sT E^\ast\op_M\sT^\ast E^\ast$ on the dual $E^\ast$ to a vector bundle $\zt:E\ra M$.
If this almost Dirac structure is integrable (Dirac), we speak about a Dirac-Lie algebroid. The bundle $E$
plays the role of the bundle of kinematic configurations (quasi-velocities), while the bundle $E^\ast$ -- the
role of the phase space. This setting is totally intrinsic and does not distinguish between regular and
singular Lagrangians. The constraints are part of the framework, so the general approach does not change when
nonholonomic constraints are imposed, and produces the (implicit) Euler-Lagrange and Hamilton equations in an
elegant geometric way. The scheme includes all important cases of Lagrangian and Hamiltonian systems, no
matter if they are with or without constraints, autonomous or non-autonomous etc., as well as their
reductions; in particular, constrained systems on Lie algebroids. we prove also some basic facts about the
geometry of Dirac and Dirac-Lie algebroids.

\bigskip\noindent
\textit{MSC 2010: 37J05, 70G45, 70F25, 57D17, 70H45, 70H03, 70H25, 17B66.}

\medskip\noindent
\textit{Key words: variational calculus, geometrical mechanics, nonholonomic constraint, Euler-Lagrange
equation, Dirac structure, Lie algebroid.}
\end{abstract}
\section{Introduction}
The concept of {\it Dirac structure}, proposed by Dorfman \cite{Do} in the Hamiltonian framework of integrable
evolution equations and defined in \cite{Co} as a subbundle of the Whitney sum $\sT N\oplus_N\sT^\ast N$ of
the tangent and the cotangent bundle (the {\it extended tangent} or the {\it Pontryagin bundle}) satisfying
certain conditions, was thought-out as a common generalization of Poisson and presymplectic structures. It was
designed also to deal with constrained systems, including constraints induced by degenerate Lagrangians, as
was investigated by Dirac \cite{Di}, which is the reason for the name.

The need of extending the geometrical tools of the Lagrangian formalism from tangent bundles to Lie algebroids
was caused by the fact that reductions usually move us out of the environment of the tangent bundles
\cite{CHMR} (think on the reduction to $\mathfrak{so}(3,\R)$ for the rigid body). It is similar to the
better-known situation of passing from the symplectic to the Poisson structures by a reduction in the
Hamiltonian formalism.

Note that the use of Lie algebroids and Lie groupoids for describing some systems of Analytical Mechanics was
proposed by P.~Libermann \cite{Li} and A.~Weinstein \cite{We}, and then developed by many authors, for
instance \cite{CLMMM,LMM,Mar1,M1}, making use of Lie algebroids in various aspects of Analytical Mechanics and
Classical Field Theory.

Since a Lie algebroid structure on a vector bundle $\zt:E\ra M$ can be viewed as a linear Poisson structure
$\Pi$ on the dual bundle $\zp:E^*\ra M$, a properly defined `linear' Dirac structure should be viewed as a
generalization of the concept of Lie algebroid. Linear structures of different kinds on a vector bundle can be
viewed, in turn, as associated with certain {\it double vector bundles}. The double vector bundles, introduced
in \cite{Pr1,Pr2} (see also \cite{KU,GR}) as manifolds with two `compatible' vector bundle structures, have
been successfully applied in \cite{GU3,GU2} to geometric formalisms of Analytical Mechanics, including
nonholonomic constraints \cite{GG,GLMM}. To be more precise, note first that canonical examples of {double
vector bundles} are: the tangent $\sT E$, and the cotangent bundle $\sT^\ast E$ of the vector bundle $E$. The
double vector bundles
$$\xymatrix{
\sT^\ast E^\ast\ar[rr]^{\sT^\ast\zp} \ar[d]_{\zp_{E^\ast}} && E\ar[d]^{\zt} \\
E^\ast\ar[rr]^{\zp} && M } \qquad {,}\qquad \xymatrix{
\sT^\ast E\ar[rr]^{\sT^\ast\zt} \ar[d]_{\zt_{E^\ast}} && E^\ast\ar[d]^{\zp} \\
E\ar[rr]^{\zt} && M }
$$
are canonically isomorphic (cf. \cite{KU,Ur}). In particular, all arrows correspond to vector bundle
structures and all pairs of vertical and horizontal arrows are vector bundle morphisms. Double vector bundles
have been recently characterized \cite{GR} in a simple way as two vector bundle structures whose Euler vector
fields commute.

In \cite{GU2,GU3}, a Lie algebroid (and its generalizations) on $E$ has been viewed as a double vector bundle
morphism
\be\label{tt}\varepsilon:\sT^\ast E\rightarrow \sT E^\ast\ee
covering the identity on $E^\ast$. This is because the linearity of a bivector field (e.g. a Poisson tensor)
$\zP_\ze$ on the dual bundle $E^\ast$ can be geometrically expressed as respecting the double vector bundle
structures by the induced vector bundle morphism
\be\label{mp}\wt{\zP}_\ze:\sT^\ast E^\ast\rightarrow \sT E^\ast\,.\ee
We obtain $\ze$ as the composition of the canonical isomorphism of double vector bundles $\cR_\zt \colon
\sT^\*E \ra \sT^\* E^\*$ with $\wt{\zP}_\ze$, $\ze=\wt{\zP}_\ze\circ\cR_\zt$.

An application of this approach to Analytical Mechanics, in which $\tau: E\rightarrow M$ plays the role of
kinematic configurations, is based on some ideas of Tulczyjew and Urba\'nski \cite{Tu1,Tu3,TU}.

Note that we can represent the morphism (\ref{mp}) of double vector bundles  by its graph $D_\ze$ in the
Whitney sum bundle
\be\label{dd}\cT E^\ast=\sT E^\ast\oplus_{E^\ast}\sT^\ast E^\ast\,.\ee
The {\it Pontryagin bundle} $\cT E^\ast$ is canonically a double vector bundle: over $E^\ast$ and over $\sT
M\oplus_M E$, and the fact that $\ze$ is a morphism means that $D_\ze$ is a double vector subbundle. Moreover,
since $D_\ze$ is the graph of a Poisson tensor (in the case when $E$ is a standard Lie algebroid), the
subbundle $D_\ze$ is a Dirac structure on $E^\ast$. This immediately leads to a generalization of the concept
of Lie algebroid: we replace the graph $D_\ze$ with any Dirac structure $D$ on $E^\ast$ which is linear, i.e.,
which is a double vector subbundle of $\cT E^\ast$. Such an object we will cal a {\it Dirac-Lie algebroid}.

As was observed already in \cite{GGU3}, the construction of phase dynamics associated with a given Lagrangian
does not use the fact that the bivector field $\Pi_\ze$ is Poisson (which, on the other hand, induces nice
properties of the dynamics), so we will use also almost Dirac structures, imposing no integrability
assumptions. Thus, a {\it Dirac algebroid} on $E$ will be a linear almost Dirac structure on $E^\ast$. We
introduce also affine analogs of Dirac and Dirac-Lie algebroids.

The main applications we propose go back again to Analytical Mechanics. To some extent, our concepts are
similar to that of \cite{YM1}, where (almost) Dirac structures have been used in description of `implicit'
Lagrangian systems. Our approach, however, we find much more general (we work with arbitrary vector bundles)
and much simpler. This is because we obtain `implicit Lagrangian systems' (in fact both: implicit phase
dynamics and implicit Euler-Lagrange equations), as well as implicit Hamilton equations, just composing
relations, instead of working with somehow artificial concept of {\it partial vector field}. This generality
allows us to cover a large variety of Lagrangian and Hamiltonian systems, including reduced systems,
nonholonomic or vakonomic constraints, and time-dependent systems, with no regularity assumptions on
Lagrangian nor Hamiltonian.

The paper is organized as follows. In section 2 we recall basic facts concerning double vector bundle approach
to Lie algebroids and their generalizations. Dirac algebroids, Dirac-Lie algebroids, and their affine
counterparts are introduced in section 3, together with main examples. In section 4 we investigate closer the
structure of Dirac algebroids, finding a short exact sequence of Lie algebroids associated with a Dirac-Lie
algebroid and providing a local form of Dirac algebroids. Section 5 is devoted to inducing new Dirac
algebroids by means of nonholonomic constraints. In section 6 we present the general schemes, based on Dirac
algebroids, for Lagrangian and Hamiltonian formalisms. We end up with a number of examples in section 7 and
concluding remarks in section 8.

\section{Lie algebroids as double vector bundle morphisms\label{S1}}
We start with recalling basic facts and introducing some notation.

Let $M$ be a smooth manifold and let $(x^a), \ a=1,\dots,n$, be a coordinate system in $M$. We denote with
$\zt_M \colon \sT M \rightarrow M$ the tangent vector bundle and by $\zp_M \colon \sT^\* M\rightarrow M$ the
cotangent vector bundle. We have the induced (adapted) coordinate systems, $(x^a, {\dot x}^b)$ in $\sT M$ and
$(x^a, p_b)$ in $\sT^\* M$.
        More generally, let $\zt\colon E \rightarrow M$ be a vector bundle and let $\zp
\colon E^\* \rightarrow M$ be the dual bundle.
  Let $(e_1,\dots,e_m)$  be a basis of local sections of $\zt\colon
E\rightarrow M$ and let $(e^{1}_*,\dots, e^{m}_*)$ be the dual basis of local sections of $\zp\colon
E^\*\rightarrow M$. We have the induced coordinate systems:
    $(x^a, y^i),  y^i=\zi(e^{i}_*)$, {in} $E$, and
    $(x^a, \zx_i), \zx_i = \zi(e_i)$, {in} $E^\*$ ,
    where the linear functions  $\zi(e)$ are given by the canonical pairing
    $\zi(e)(v_x)=\la e(x),v_x\ran$. In this way we get local coordinates
    \beas
    (x^a, y^i,{\dot x}^b, {\dot y}^j ) &  \text{in} \ \sT E,\quad
    (x^a, \zx_i, {\dot x}^b, {\dot \zx}_j) &  \text{in} \ \sT E^\* ,\\
    (x^a, y^i, p_b, \zp_j) &  \text{in}\ \sT^\*E,\quad
    (x^a, \zx_i, p_b, \zf^j) &  \text{in}\ \sT^\* E^\* .
    \eeas

The cotangent bundles $\sT^\*E$ and $\sT^\*E^\*$ are examples of so called {\it double vector bundles}. They
are fibred over $E$ and $E^\*$ and canonically isomorphic, with the isomorphism $\cR_\zt \colon \sT^\*E
\longrightarrow \sT^\* E^\*$, being simultaneously an anti-symplectomorphism  (cf. \cite{KU,GU2}). In local
coordinates, $\cR_\zt$ is given by
    \be\label{iso1}\cR_\zt(x^a, y^i, p_b, \zp_j) = (x^a, \zp_i, -p_b,y^j).
                              \ee
This means that we can identify coordinates $\zp_j$ with $\zx_j$, coordinates $\zf^j$ with $y^j$, and use the
coordinates $(x^a, y^i, p_b, \zx_j)$ in $\sT^\ast E$ and the coordinates $(x^a, \zx_i, p_b,y^j)$ in $\sT^\ast
E^\ast$ in full agreement with (\ref{iso1}). According to \cite{GR}, the double vector bundle structure is
completely characterized by a pair of commuting Euler vector fields defining the two vector bundle structures
(or by the pair of the corresponding families of homotheties). In local coordinates the Euler vector fields on
$\sT^\ast E^\ast$ read
\be\label{EVF} \nabla_{\sT^\ast
E}^E=p_b\pa_{p_b}+\zx_i\pa_{\zx_i}\,,\quad \nabla_{\sT^\ast E}^{E^\ast}=p_b\pa_{p_b}+y^i\pa_{y^i}\,.
\ee
Double vector (and vector-affine) bundles will play an important role in our concepts and we refer to
\cite{Ma, GR,GRU,KU,Ur} for the general theory.

It is well known that Lie algebroid structures on a vector bundle $E$ correspond to linear Poisson tensors on
$E^\*$. A 2-contravariant tensor $\zP$ on $E^\*$ is called {\it linear} if the corresponding mapping
$\widetilde{\zP} \colon \sT^\* E^\* \rightarrow \sT E^\*$ induced by the contraction,
$\wt{\zP}(\zn)=i_\zn\zP$, is a morphism of double vector bundles. One can equivalently say that the
corresponding bracket of functions is closed on (fiber-wise) linear functions. The commutative diagram
$$\xymatrix{
\sT^\ast E^\ast\ar[r]^{\widetilde\Pi}  & \sT E^\ast \\
\sT^\ast E\ar[u]_{\cR_\tau}\ar[ur]^{\ze} & },
$$
describes a one-to-one correspondence between linear 2-contravariant tensors $\zP_\ze$ on $E^\*$ and morphisms
$\ze$ (covering the identity on $E^\*$) of the following double vector bundles (cf. \cite{KU, GU2}):

\be\xymatrix{
 & \sT^\ast E \ar[rrr]^{\varepsilon} \ar[dr]^{\pi_E}
 \ar[ddl]_{\sT^\ast\tau}
 & & & \sT E^\ast\ar[dr]^{\sT\pi}\ar[ddl]_/-20pt/{\tau_{E^\ast}}
 & \\
 & & E\ar[rrr]^/-20pt/{\zr}\ar[ddl]_/-20pt/{\tau}
 & & & \sT M \ar[ddl]_{\tau_M}\\
 E^\ast\ar[rrr]^/-20pt/{id}\ar[dr]^{\pi}
 & & & E^\ast\ar[dr]^{\pi} & &  \\
 & M\ar[rrr]^{id}& & & M &
}\label{F1.3}
\ee
In local coordinates, every  such $\ze$ is of the form \be\label{F1.4} \ze(x^a,y^i,p_b,\zx_j) = (x^a, \zx_i,
\zr^b_k(x)y^k, c^k_{ij}(x) y^i\zx_k + \zs^a_j(x) p_a)
\ee
(summation convention assumed) and it corresponds to the linear tensor $$ \zP_\ze =c^k_{ij}(x)\zx_k
\partial _{\zx_i}\otimes \partial _{\zx_j} + \zr^b_i(x) \partial _{\zx_i}
\otimes \partial _{x^b} - \zs^a_j(x)\partial _{x^a} \otimes
\partial _{\zx_j}.
$$
The morphism (\ref{F1.3}) of double vector bundles covering the identity on $E^\*$ has been called an {\it
algebroid} in \cite{GU2}. We will consider only {\it skew algebroids}, i.e., algebroids $\ze$ for which the
tensor $\zP_\ze$ is skew-symmetric, i.e., is a bivector field. If $\zP_\ze$ is a Poisson tensor, we deal with
a {\it Lie algebroid}. The relation to the canonical definition of Lie algebroid is given by the following
theorem (cf. { \cite{GU3, GU2}}).

\begin{theo}\label{bracket}
A skew algebroid structure $(E,\ze)$ can be equivalently defined as a skew-symmetric bilinear bracket $[\cdot
,\cdot]_\ze $ on the space $\Sec(E)$ of sections of $\zt\colon E\rightarrow M$, together with a vector bundle
morphisms $\zr \colon E\rightarrow \sT M$ (called the {\it anchor}), such that
$$ [X,fY]_\ze = \zr(X)(f)Y+f [X,Y]_\ze
$$
         for $f \in \cC^\infty (M)$, $X,Y\in \Sec(E)$.
The bracket and the anchor are related to the  bracket
$\{\zf,\psi\}_{\zP_\ze}=\la\zP_\ze,\xd\zf\ot\xd\psi\ran$ in the algebra of functions on $E^\ast$, associated
with the bivector field $\zP_\ze$, by the formulae
\beas
        \zi([X,Y]_\ze)&= \{\zi(X), \zi(Y)\}_{\zP_\ze},  \\
        \zp^\*(\zr(X)(f))       &= \{\zi(X),
        \zp^\*f\}_{\zP_\ze}\,,
                                                   \eeas
where $\zi(X)$ is the linear function on $E^\ast$ associated with the section $X$ of $E$.
\end{theo}

\section{Dirac algebroids and affine Dirac algebroids}
Let $N$ be a smooth manifold. There is a natural symmetric pairing $(\cdot|\cdot)$ on the vector bundle $\cT
N=\sT N\oplus_N\sT^\ast N$ (called sometimes the {\it Pontryagin bundle}) given by
$$(X_1 + \za_1|X_2 + \za_2) = \frac{1}{2}\left(\za_1(X_2) + \za_2(X_1)\right)\,,
$$
for all sections $X_i+\za_i$, $i=1,2$, of $\cT N=\sT N\oplus_N\sT^\ast N$. Furthermore, the space $\Sec(\cT
N)$ of smooth sections of $\cT N$ is endowed with the Courant-Dorfman bracket,
\be\label{CD}
\lna X_1 + \za_1,X_2 + \za_2\rna = [X_1,X_2] + \cL_{X_1}\za_2 - i_{X_2}\xd \za_1\,,
\ee
where $[\cdot,\cdot]$ is the Lie bracket of vector fields, $\cL_X$ is the Lie derivative along the vector
field $X$, and $i_X$ is the contraction (inner product) with $X$. An {\it almost Dirac structure (or bundle)}
on the smooth manifold $N$ is a subbundle $D$ of $\cT N$ which is maximally isotropic with respect to the
symmetric pairing $(\cdot| \cdot)$. If additionally the space of sections of $D$ is closed under the
Courant-Dorfman bracket, we speak about a {\it Dirac structure (or bundle)} \cite{Co,Do}.

Standard examples of almost Dirac structures are the graphs
\beas\Graph(\zP)&=&\{ X_p+\za_p\in \cT_pN:p\in N\,,\ X_p=\wt{\zP}(\za_p)\}\,,\\
\Graph(\zw)&=&\{ X_p+\za_p\in \cT_pN:p\in N\,,\ \za_p=\wt{\zw}(X_p)\}\,,
\eeas
of bivector fields $\zP$ or 2-forms $\zw$ viewed as vector bundle morphisms,
\beas\wt{\zP}&:&\sT^\ast N\ra\sT N\,,\
\wt{\zP}(\za_p)=i_{\za_p}\zP(p)\,,\\
\wt{\zw}&:&\sT N\ra\sT^\ast N\,,\ \wt{\zw}(X_p)=-i_{X_p}\zw(p)\,.
\eeas
These graphs are actually Dirac structures if and only if $\zP$ is a Poisson tensor and $\zw$ is a closed
2-form, respectively.

\begin{rem} A vector subbundle of a vector bundle over $N$ is often
understood as a vector bundle over the whole base manifold $N$. It is however clear by many reasons (see e.g.
\cite[Theorem 2.3]{GR}) that we must consider also vector subbundles supported on submanifolds of $N$.
Throughout this paper the term {\it vector subbundle} always means a subbundle of the original vector bundle
supported on a submanifold $N_0\subset N$. In this sense, our definitions of almost Dirac and Dirac structure
are slightly more general than those usually available in the literature. By `being closed' with respect to
the bracket we clearly mean that the bracket of any two sections of $\cT N$, extending sections of $D$, does
not depend over $N_0$ on the extensions chosen and gives a section extending a section of $D$. This uniquely
defines a bracket on sections of $D$ which is known to be a Lie algebroid bracket.
\end{rem}
Since the projection $\pr_{\sT N}:\cT N\ra \sT N$ is the left anchor for the Courant-Dorfman bracket, i.e.,
\be\label{anchor}\lna X_1+\za_1,f(X_2+\za_2)\rna=f\lna X_1+\za_1,X_2+\za_2\rna+X_1(f)(X_2+\za_2)\,,
\ee
it is a straightforward observation that the bracket of extensions of sections of a subbundle $D$, supported
on a submanifold $N_0$ of $N$, does not depend on the extensions if and only if
\be\label{1ic}\pr_{\sT N}(D)\subset \sT N_0\,.
\ee
Indeed, if $f$ is 0 on $N_0$, by (\ref{anchor}) $X_1(f)$ must be 0 on $N_0$ for any section $X_1+\za_1$ which
belongs to $D$ along $N_0$. The condition (\ref{1ic}) we will call the {\it first integrability condition} for
the Dirac-Lie algebroid. Under this condition the Courant-Dorfman bracket restricts to
\be\label{restr}\lna\cdot,\cdot\rna_{D}:\Sec(D)\ti\Sec(D)\ra\Sec(\cT N)\,.
\ee
 Then, the {\it second integrability condition} says that $\lna\cdot,\cdot\rna_{D}$ takes values in $\Sec(D)$:
\be\label{2i}
\lna\cdot,\cdot\rna_{D}:\Sec(D)\ti\Sec(D)\ra\Sec(D)\subset\Sec((\cT N)_{|N_0})\,,
\ee
which, according to (\ref{anchor}) and (\ref{1ic}), is sufficient to be checked on a generating set of
sections of $D$:
\be\label{2ic}
\lna\zs_k,\zs_l\rna_D\in\Sec(D)\ \text{for}\ \{\zs_i\}\subset\Sec(D)\ \text{generating}\ D\,.
\ee
By definition, an almost Dirac structure is a Dirac structure if and only if it satisfies both the
integrability conditions, (\ref{1ic}) and (\ref{2ic}).

\begin{rem} Suppose that an almost Dirac structure $D$ satisfies the first integrability condition,
i.e., the Courant-Dorfman bracket $\lna\cdot, \cdot\rna_D$ of sections of $D$ supported on $N_0$ is well
defined. If we have chosen a subbundle $K$ of $\cT N$ complementary to $D$ over $N_0$, we can define the
bracket
\be\label{br}\lna\cdot, \cdot\rna_D^K:\Sec(D)\ti\Sec(D)\ra\Sec(D)
\ee
by projecting the value of $\lna\cdot,\cdot\rna_{D}$ onto $\Sec(D)$ along $K$. Of course, if $D$ is a Dirac
structure, $\lna\cdot, \cdot\rna_D^K$ does not depend on the choice of $K$ and is just the Lie algebroid
bracket on sections of $D$.
\end{rem}

In Geometric Mechanics there is often a need to use affine bundles and affine versions of algebroids
\cite{MMS1,GGU1,GGU2,GGU4,GU,IMMS,IMPS} ({\it affgebroids}  in the terminology introduced in
\cite{GGU1,GGU2}). We will use the following concept.

\begin{deff} Let $A$ be an affine subbundle of a Lie algebroid $E\ra M$ with the bracket $[\cdot,\cdot]$
and the anchor $\zr:E\ra\sT M$, supported on a submanifold $S\subset M$. Let $V=\sv(A)$ be its model vector
bundle viewed as a vector subbundle of $E$. We call $A$ an {\it affine Lie subalgebroid in $E$}, if the
brackets of sections of $A$ lie in $\Sec(\sv(A))$, i.e., $\zr(A)\subset \sT S$ (thus the bracket of sections
of $A$ is well defined over $S$) and $[\zs,\zs']\in\Sec(V)$ for all $\zs,\zs'\in\Sec(A)$.
\end{deff}
\noindent For a more extensive treatment of brackets on affine bundles we refer to \cite{GGU1,GGU2} (see also
\cite{MMS1,IMMS,IMPS}).

To consider also affine versions of (almost) Dirac structures, we propose the following (compare
\cite{GGU1,GGU2}).
\begin{deff} An {\it affine almost Dirac structure} on a manifold $N$ is an affine
subbundle $D$ of $\cT N$, supported on a submanifold $N_0$ of $N$, whose model vector bundle $\sv(D)\subset\cT
N$ (canonically represented by a subbundle of $\cT N$) is an almost Dirac structure on $N$. An affine almost
Dirac structure is called {\it affine Dirac structure}, if the Courant-Dorfman bracket of sections of $D$
makes sense (like the analogous concept for Dirac-Lie algebroids) and takes values in the set of sections of
$\sv(D)$, i.e., (\ref{1ic}) is satisfied, so that (\ref{restr}) is well defined and
\be\label{2ica}
\lna\cdot,\cdot\rna_{D}:\Sec(D)\ti\Sec(D)\ra\Sec(\sv(D))\subset\Sec(\cT N)\,.
\ee
\end{deff}
\noindent The following is straightforward.
\begin{prop}\label{p0} If $D$ is an affine Dirac structure, then $\sv(D)$ is a Dirac structure.
\end{prop}

Let now $F$ be a vector bundle over a manifold $M$. Since both, $\sT F$ and $\sT^\ast F$, are canonically
double vector bundles, their Whitney sum carries a structure of canonical double vector bundle as well. From the general
theory we easily derive the following  (cf. \cite{KU,GR}).
\begin{theo}\label{core} If $F$ is a vector bundle over $M$, its Pontryagin bundle
$\cT F=\sT F\oplus_F\sT^\ast F$, canonically isomorphic to $\sT F\oplus_F\sT^\ast F^\ast$, is also canonically
a double vector bundle structure with two compatible vector bundle structures: $\zt_1:\cT F\ra F$ and
$\zt_2:\cT F\ra\sT M\oplus_M F^\ast$.

The {\em core bundle} of $\cT F$, i.e., a vector bundle over $M$ being the intersection of the kernels of the
both projections, is in this case canonically isomorphic to $\sT^\ast M\oplus_M F$. Moreover, the fibration
$$(\zt_1,\zt_2):\cT F\ra F\oplus_M\sT M\oplus_M F^\ast$$ is an affine bundle modeled on the pull-back core bundle,
i.e., the core bundle $\sT^\ast M\oplus_M F$ over $M$ pulled-back to $F\oplus_M\sT M\oplus_M F^\ast$ {\it via}
the canonical projection $F\oplus_M\sT M\oplus_M F^\ast\ra M$.
\end{theo}
\begin{deff}
We call a submanifold $D$ of a double vector bundle its {\it double vector subbundle}, if $D$ is a subbundle
for each of the two vector bundle structures.
\end{deff}
Following the ideas of \cite{GR}, one can easily prove that this means  that the two Euler vector fields
defining the double vector bundle structure are tangent to $D$. One can also equivalently say that $D$ is
invariant with respect to both commuting families of homotheties defined by the two vector bundle structures
(cf. \cite{GR}).

\begin{prop}\label{p1} Let $D$ be a double vector subbundle of a double vector bundle
\begin{equation}\label{dvb}
\xymatrix{ K\ar[r]^{\zt_2} \ar[d]^{\zt_1} & K_2\ar[d]^{\zt_1'} \\
K_1\ar[r]^{\zt_2'} & M }
\end{equation}
Then, $D$ inherits a double vector bundle structure with projections onto  vector bundles $S_i=\zt_i(D)$,
$i=1,2$, where $S_i$ is a vector subbundle of $K_i$.
\end{prop}
\begin{proof} It is esy to see that the homothety $h^1_t$, being
the multiplication of vectors of the bundle $\zt_1:K\ra K_1$ by $t\in\R$, coincides on $K_2$ with the
homothety of the vector bundle $K_2\ra M$. The submanifold $D$, being $h^1_t$-invariant, has the base
$S_2\subset K_2$ which is $h^1_t$-invariant, thus is a vector subbundle of $K_2\ra M$ \cite[Theorem 2.3]{GR}.
\end{proof}

\begin{deff} A {\it Dirac algebroid} (resp., {\it Dirac-Lie algebroid}) {\it
structure} on a vector bundle $E$ is an almost Dirac (resp., Dirac) subbundle $D$ of $\cT E^\ast$ being a
double vector subbundle, i.e., $D$ is not only a subbundle of $\zt_1:\cT E^\ast\ra E^\ast$ but also a vector
subbundle of the vector bundle $\zt_2:\cT E^\ast\ra\sT M\oplus_M E$.
\end{deff}
\begin{rem} The above definition gives an analog of linearity of a Poisson or a presymplectic structure.
\end{rem}
\noindent We will consider also {\it affine Dirac algebroids} ({\it Dirac affgebroids} in short).
\begin{deff} An {\it affine Dirac algebroid} on a vector bundle $E$ is an affine
subbundle $D$ of $\cT E^\ast$ whose model vector bundle $\sv(D)\subset\cT E^\ast$ (represented by vertical
vectors tangent to the fibers of $D$) is a Dirac algebroid. An affine Dirac algebroid is called {\it affine
Dirac-Lie algebroid}, if $D$ is an affine Dirac structure, i.e., if the Courant-Dorfman bracket of sections of
$D$ is a section of $\sv(D)$.
\end{deff}
\noindent According to proposition \ref{p0}, if $D$ is an affine Dirac-Lie algebroid, then $\sv(D)$ is a
Dirac-Lie algebroid.

\begin{rem} We can consider as well other affine types of Dirac structures,
defined on affine or special affine bundles, by considering vector-affine bundles of different types (see e.g.
\cite{GRU}), but we skip these considerations here in order not to multiply technical difficulties.
\end{rem}

\medskip
In view of proposition \ref{p1}, a Dirac algebroid $D\subset\cT E^\ast$ projects onto two vector subbundles:
$\Ph_D=\zt_1(D)\subset E^\ast$ and $\Ve_D=\zt_2(D)\subset\sT M\oplus_M E$, both based on a submanifold $M_D$
of $M$, giving rise to a single projection,
\be\label{proj} \zt^D=(\zt_1^D,\zt_2^D):D\ra\Ph_D\op_{M_D}\Ve_D\subset E^\ast\op_M(\sT M\op_M E)\,,
\ee
which, according to theorem \ref{core}, is an affine bundle modeled on the core $\cC_D$ of $D$ pulled-back to
$\Ph_D\op_{M_D}\Ve_D$, i.e., on $(\Ph_D\op_{M_D}\Ve_D)\ti_{M_D}\cC_D$. Note that the core $\cC_D$ is a
subbundle (supported on $M_D$) of $\sT^\ast M\oplus_M E^\ast$ -- the core of the double vector bundle $\cT E$.

The first component in $\Ph_D\op_{M_D}\Ve_D$ we will call the {\it phase bundle} and the second -- the {\it
anchor relation} (or the {\it velocity bundle}) of the Dirac algebroid $D$. The anchor relation is just a
linear relation between vectors of $E$ (`quasi-velocities') and  vectors tangent to $M$ (`actual velocities')
and gives rise to the {\it anchor map}
\be\label{vam}
\zr_D:\sT M\op_M E\supset\Ve_D\ra\sT M_D
\ee
being the projection onto the first summand.

\medskip
To express linearity of an almost Dirac (or Dirac) subbundle of $\cT E^\ast$ in a more explicit way, consider
adapted coordinates $(x^a,\zx_i,\dot{x}^b,\dot{\zx}_j,p_c,y^k)$ on $\cT E^\ast$. The two commuting Euler
vector fields are:
$$\nabla_1=p_a\pa_{p_b}+\dot{\zx}_j\pa_{\dot{\zx}_j}+y^i\pa_{y^i}+\dot{x}^b\pa_{\dot{x}^b}\,,$$
corresponding to the vector bundle structure over $E^\ast$ with coordinates $(x,\zx)$, and
$$\nabla_2=p_a\pa_{p_b}+\zx_i\pa_{\zx_i}+\dot{\zx}_j\pa_{\dot{\zx}_j}\,,$$
corresponding to vector bundle structure over $\sT M\oplus_M E$ with coordinates $(x,\dot{x},y)$. The
corresponding homotheties read
\bea\label{ho1} h^1_t(x^a,\zx_i,\dot{x}^b,\dot{\zx}_j,p_c,y^k)&=&(x^a,\zx_i,t\dot{x}^b,t\dot{\zx}_j,tp_c,ty^k)\,,\\
\label{ho2}h^2_s(x^a,\zx_i,\dot{x}^b,\dot{\zx}_j,p_c,y^k)&=&(x^a,s\zx_i,\dot{x}^b,s\dot{\zx}_j,sp_c,y^k)\,,
\eea
and a linear almost Dirac subbundle in $\cT E^\ast$ (Dirac algebroid) should be invariant with respect to both
sets of homotheties. Note that the canonical symmetric pairing is represented by the quadratic function
$Q(x^a,\zx_i,\dot{x}^b,\dot{\zx}_j,p_c,y^k)=p_a\dot{x}^a+y^i\dot{\zx}_i$ which vanishes on Dirac algebroids.

\medskip

\begin{example}\label{e1} The graph of any linear bivector field
$$ \zP =\frac{1}{2}c^k_{ij}(x)\zx_k
\partial _{\zx_i}\we \partial _{\zx_j} + \zr^b_i(x) \partial _{\zx_i}
\wedge \partial _{x^b}\,,
$$
where $c^k_{ij}=-c^k_{ji}$, is a Dirac algebroid:
$$\Graph(\zP)=\{(x^a,\zx_i,\dot{x}^b,\dot{\zx}_j,p_c,y^k):
\dot{x}^b=\zr^b_k(x)y^k\,,\ \dot{\zx}_j=c^k_{ij}(x) y^i\zx_k - \zr^a_j(x) p_a\}\,.$$ It is clear that $Q$
vanishes on $\Graph(\zP)$. This graph is a double vector subbundle, since the constraint functions
\be\label{kk}\dot{x}^b-\zr^b_k(x)y^k\,,\ \dot{\zx}_j-c^k_{ij}(x) y^i\zx_k + \zr^a_j(x) p_a
\ee
are homogeneous with respect to the Euler vector fields $\nabla_1,\nabla_2$. The phase bundle is here $E^\ast$
and the anchor relation is actually the graph of the vector bundle morphism $\zr:E\ra\sT M$ (the anchor map)
given in local coordinates by $\zr(x^a,y^i)=(x^a,\zr_i^b(x)y^i)$. This means that skew-algebroids are
particular examples of Dirac algebroids. The Dirac algebroids of this form, associated with a bivector field
$\zP$, we will call {\it $\zP$-graph Dirac algebroids} on $E$ and denote $D_\zP$. The Dirac algebroid $D_\zP$
is a Dirac-Lie algebroid if and only if $\zP$ is a Poisson tensor, i.e., if and only if we deal with a Lie
algebroid.
\end{example}
\begin{example}\label{e2} The graph of any linear 2-form
$$ \zw=\frac{1}{2}c^k_{ab}(x)\zx_k\xd x^a\we\xd x^b + \zr_b^i(x) \xd{\zx_i} \wedge \xd{x^b}\,,
$$
where $c^k_{ab}=-c^k_{ba}$, is a Dirac algebroid:
$$\Graph(\zw)=\{(x^a,\zx_i,\dot{x}^b,\dot{\zx}_j,p_c,y^k):
y^i=\zr^i_a(x)\dot{x}^a\,,\ p_a=c^k_{ab}(x)\zx_k\dot{x}^b - \zr_a^i(x)\dot{\zx}_i\}\,.$$ It is clear that $Q$
vanishes on $\Graph(\zw)$. This graph is a double vector subbundle, since the constraint functions
\be\label{kkk}y^i-\zr^i_a(x)\dot{x}^a\,,\ p_a-c^k_{ab}(x)\zx_k\dot{x}^b +
\zr_a^i(x)\dot{\zx}_i
\ee
are homogeneous with respect to the Euler vector fields $\nabla_1,\nabla_2$. The phase bundle is here $E^\ast$
and the anchor relation is in fact the graph of the vector bundle morphism $\zr:\sT M\ra E$ given in local
coordinates by $\zr(x^a,\dot{x}^b)=(x^a,\zr^i_b(x)\dot{x}^b)$. The Dirac algebroids of this form, associated
with a 2-form $\zw$, we will call {\it $\zw$-graph Dirac algebroids} and denote $D_\zw$. The Dirac algebroid
$D_\zw$ is a Dirac-Lie algebroid ({\it presymplectic Dirac-Lie algebroid}) if and only if $\zw$ is closed.
\end{example}
\begin{example}\label{e3} The {\it canonical Dirac-Lie algebroid}
$D_M=D_{\zP_M}=D_{\zw_M}$, corresponding to the canonical Lie algebroid $E=\sT M$, belongs to the both above
types. It is associated with the canonical symplectic form $\zw_M$ on $E^\ast=\sT^\ast M$ and, simultaneously,
to the canonical Poisson tensor $\zP_M=\zw_M^{-1}$ on $\sT^\ast M$. In our local coordinates, the equations
defining $D_M$ are
$$\dot{x^a}=y^a\,,\ \dot{\zx_b}=-p_b\,.$$
\end{example}
\begin{example}\label{e5}
Suppose we have a Dirac (Dirac-Lie) algebroid $D$ on $E\ra M$. Let us consider the extension $E_0=E\ti\R$ as a
vector bundle over $M_0=M\ti\R$ in the obvious way. Then, $E_0^\ast=E^\ast\ti\R$, $\sT E_0^\ast=\sT
E^\ast\ti\sT \R$, and $\sT^\ast E_0^\ast=\sT^\ast E^\ast\ti\sT^\ast\R$. The subbundle $D_0=D\ti A_0$ in $\cT
E_0^\ast=\cT E^\ast\ti\cT\R$, where $A_0$ is the affine subbundle in $\cT\R$ defined by the constraint
$\dot{x}_0=1$ in the natural coordinates $(x_0,\dot{x_0},p_0)$ on $\cT\R$, is an affine Dirac (Dirac-Lie)
algebroid on $E_0$.
\end{example}

\section{The structure of a Dirac algebroid}

Let us start this paragraph with recalling that any section $\zs:N\ra F$ of a vector bundle $F\ra N$
(actually, of any fibration) is uniquely determined by its image $\zs(N)$ -- a submanifold of $F$. We will
denote this submanifold by $[\zs]$.
\begin{deff} Let $K$ be a double vector bundle (\ref{dvb}). We say that a section $\wt{\zs}:K_1\ra K$ {\it projects} on the section $\zs:M\ra K_2$,
if $\zt_2$ projects $[\wt{\zs}]$ onto $[\zs]$. We will write $\wt{\zs}^{\zt_2}=\zs$ and call such $\wt{\zs}$
{\it projectable}.

We say that a section $\wt{\zs}:K_1\ra K$ is {\it suitable}, if $[\wt{\zs}]$ is a vector subbundle of the
vector bundle $\zt_2:K\ra K_2$.
\end{deff}
It is easy to see the following
\begin{theo}
Any suitable section $\wt{\zs}$ is projectable and $[\wt{\zs}^{\zt_2}]$ is the image under $\zt_2\wt{\zs}$ of the
zero-section $0_{K_1}$ of $\zt_1$. Moreover, the set of suitable sections, $\Suit(K)$, is canonically a
$C^\infty(M)$-module and the module morphism $[\zt_2]:\wt{\zs}\mapsto\wt{\zs}^{\zt_2}$ is an epimorphism onto
$\Sec(K_2)$.
\end{theo}
Suitable sections which project on the zero-section of the bundle $K_1$ we will call {\it 0-suitable}. So the
set $\Suit_0(K)$ of 0-suitable sections is the kernel of the map $[\zt_2]:\Suit(K)\ra\Sec(K_2)$. A standard
argument shows that that the $\cim$-modules $\Suit(K)$ and $\Suit_0(K)$ are the modules of sections of certain
vector bundles over $M$, $\suit(K)$ and $\suit_0(K)$, respectively, but we will not go into details here.

\medskip
All this can be applied to the situation of the Pontryagin bundle over the vector bundle $E^*$,
\begin{equation}\label{dvbp}
\xymatrix{ \cT E^\ast\ar[rr]^{\zt_2} \ar[d]^{\zt_1} && \sT M\oplus_M E\ar[d]^{\zt_M\ti\zt} \\
E^*\ar[rr]^{\zp} && M }
\end{equation}
and easily explained in our standard local coordinates $(x^a,\zx_i,\dot{x}^b,\dot{\zx}_j,p_c,y^k)$. The image
of a section $\wt{\zs}$ of $\zt_1$ consists of points
$$\left(x^a,\zx_i,\dot{x}^b(x,\zx),\dot{\zx}_j(x,\zx),p_c(x,\zx),y^k(x,\zx)\right)\in \cT E^\ast\,.$$
This section is projectable if and only if the coefficients $\dot{x}^b$ and $y^k$ depend on $x$ only,
$$\dot{x}^b=\dot{x}^b(x)\,,\ y^k=y^k(x)\,,$$ thus $\wt{\zs}$ projects onto the section
$$\zs(x)=(x,\dot{x}^b(x),y^k(x))$$
of $\sT M\op_M E$. Since being a vector subbundle means exactly being a submanifold invariant with respect to
homotheties \cite{GR}, $\wt{\zs}$ is suitable if the submanifold
$$[\wt{\zs}]=\left\{\left(x,\zx_i,\dot{x}^b(x,\zx),\dot{\zx}_j(x,\zx),p_c(x,\zx),y^k(x,\zx)\right)\in \cT E^\ast\right\}$$
is invariant with respect to homotheties (\ref{ho2}), i.e.,
\beas\dot{x}^b(x,s\zx)=\dot{x}^b(x,\zx)\,,&\ y^k(x,s\zx)=y^k(x,\zx)\,,
\\ \dot{\zx}_j(x,s\zx)=s\dot{\zx}_j(x,\zx)\,,&\ p_c(x,s\zx)=s p_c(x,\zx)\,.
\eeas
As smooth homogeneous functions are linear, we get finally that $\dot{x}^b$ and $y^k$ do not depend on $\zx$
($\wt{\zs}$ is projectable) and that $\dot{\zx}_j$ and $p_c$ linearly depend on $\zx$,
\be\label{lin}\dot{\zx}_j(x,\zx)=\dot{\zx}_j^i(x)\zx_i\,,\ p_c(x,\zx)=p_c^i(x)\zx_i\,.
\ee
Recall that the section $\wt{\zs}$ is $X+\za$, where the vector field on $E^\ast$ reads
$$X=\dot{x}^b(x,\zx)\pa_{\dot{x}^b}+\dot{\zx}_j(x,\zx)\pa_{\dot{\zx}_j}$$
and the 1-form $\za$ is
$$\za=p_c(x,\zx)\xd x^c+y^k(x,\zx)\xd\zx_k\,.$$
Since linearity is measured by homogeneity with respect to the Euler vector field in the bundle, this implies
immediately the following.
\begin{theo}\label{impo}
Let $\nabla$ be the Euler vector field in the vector bundle $E^\ast$. A section $X+\za$ of $\zt_1:\cT
E^\ast\ra E^\ast$ is suitable if and only if $\cL_\nabla X=0$ and $\cL_\nabla\za=\za$.
\end{theo}
Such vector fields and 1-forms are sometimes called, with some abuse of terminology, {\it linear}. Hence,
$X+\za$ is suitable if and only if $X$ and $\za$ are {\it linear}. This allows one to identify the bundle
$\suit(\cT E^\ast)$ with $\Der(E)\op_M(\Der(E)^\ast\ot_M E)$ with $\Der(E)$ being the bundle of {\it
quasi-derivations} (or {\it derivative endomorphisms} or {\it quasi-derivations}) in $E$ (see \cite{KSM}). We
will not go into details here.

A fundamental observation is now the following.
\begin{theo}\label{impo1} If $\,\wt{\zs}_i$, $i=1,2$, are suitable sections of $\cT E^\ast$, then
$\lna\wt{\zs}_1,\wt{\zs}_2\rna$ and $\xd\left(\wt{\zs}_1|\wt{\zs}_2\right)$ are suitable. Moreover, if
$\wt{\zs}_2$ is additionally 0-suitable, then $\lna\wt{\zs}_1,\wt{\zs}_2\rna$ and
$$\lna\wt{\zs}_2,\wt{\zs}_1\rna -2\xd\left(\wt{\zs}_1|\wt{\zs}_2\right)$$
are 0-suitable.

In particular, suitable sections of $\cT E^\ast$ are closed with respect to the Courant-Dorfman bracket and
0-suitable sections form a left-ideal inside.
\end{theo}
\begin{proof} If $X_i$ and $\za_i$ are linear, $i=1,2$, then of course
$[X_1,X_2]$ and $\cL_{X_1}\za_2-i_{X_2}\za_1$, as well as $\xd\left(i_{X_1}\za_2+i_{X_2}\za_1\right)$, are
linear. To find the projection $\lna\wt{\zs}_1,\wt{\zs}_2\rna^{\zt_2}$ in coordinates, let us write
\beas\wt{\zs}_1(x,\zx)&=&{\dot{x}^{b}}(x)\pa_{x^{b}}+f_i^j(x)\zx_j\pa_{\zx_i}+{y^i}(x)\xd\zx_i
+g_a^j(x)\zx_j\xd x\,,\\
\wt{\zs}_2(x,\zx)&=&{\dot{\bar{x}}^{b}}(x)\pa_{x^{b}}+\bar{f}_i^j(x)\zx_j\pa_{\zx_i}+{\bar{y}^i}(x)\xd\zx_i
+\bar{g}_a^j(x)\zx_j\xd x\,.
\eeas
Then, direct calculations of the Courant-Dorfman bracket show that $\lna\wt{\zs}_1,\wt{\zs}_2\rna^{\zt_2}$ is
represented by the tensor
\be\label{dbracket}
\left(\dot{x}^{c}\frac{\pa \dot{\bar{x}}^{b}}{\pa{x^{c}}}-\dot{\bar{x}}^{c}\frac{\pa
\dot{{x}}^{b}}{\pa{x^{c}}}\right)(x)\pa_{x^{b}}+\left({\dot{x}^{b}}\frac{\pa {\bar{y}^i}}{\pa
{x^{b}}}-{\dot{\bar{x}}^{b}}\frac{\pa {{y}^i}}{\pa {x^{b}}}+{\bar{y}^j}f_j^i+
{\dot{\bar{x}}^{b}}g_{b}^i\right)(x)\xd\zx_i\,.
\ee
If $\dot{\bar{x}}$ and $\bar{y}$ are 0, we get 0. If $\dot{{x}}$ and ${y}$ are 0, we get
$$\left({\bar{y}^j}f_j^i+
{\dot{\bar{x}}^{b}}g_{b}^i\right)(x)\xd\zx_i=2\left(\xd(\wt{\zs}_1|\wt{\zs}_2)\right)^{\zt}_2\,.
$$
\end{proof}

It is clear that having a double vector subbundle $D$, e.g. Dirac algebroid, we can consider suitable sections
of $D$ in the same manner. As the scalar products $\left(\wt{\zs}_1|\wt{\zs}_2\right)$ vanish for sections of
a Dirac algebroid, out of theorem \ref{impo1} we can easily derive the following. Let us fix a Dirac algebroid
with an anchor relation $\Ve_D$ inducing an anchor map $\zr_D:\Ve_D\ra\sT M_D$.
\begin{theo}\label{impo2}
If $D$ is a Dirac algebroid satisfying the first-integrability condition, then the Courant-Dorfman bracket
induces on the module of suitable sections of $D$ a skew-symmetric bracket
$$\lna\cdot,\cdot\rna_D:\Suit(D)\ti\Suit(D)\ra\Suit\left((\cT E^\ast)_{|\Ph_D}\right)$$
such that
$$\lna\wt{\zs}_1,f\wt{\zs}_2\rna_D=f\lna\wt{\zs}_1,\wt{\zs}_2\rna_D+\zr_D(\wt{\zs}_1^{\zt_2})(f)\wt{\zs}_2$$
for all $f\in C^\infty(M_D)$. Moreover, if one of the sections is 0-suitable, the the resulted bracket is
0-suitable.

In the case when $D$ is a Dirac-Lie algebroid, the Courant-Dorfman bracket is a Lie algebra bracket on
$\Suit(\cT E^\ast)$  for which $\Suit_0(\cT E^\ast)$ is a Lie ideal and turns the bundles $\suit(D)$ and
$\suit_0(D)$ into Lie algebroids. Moreover, in this situation, as $\Suit(D)/\Suit_0(D)\simeq\Sec(\Ve_D)$, we
get a Lie algebroid bracket on the anchor bundle $\Ve_D$ that gives rise to a canonical short exact sequence
of Lie algebroids associated with the Dirac-Lie algebroid $D$.
$$0\longrightarrow\suit_0(D)\longrightarrow\suit(D)\longrightarrow\Ve_D\longrightarrow 0\,.$$
\end{theo}
In the case of a Lie algebroid $E$ associated with a linear Poisson structure $\zP$ on $E^\ast$ the Lie
bracket of sections of $E$ can be recognized inside the Lie algebroid on sections of $D_\zP$ as the bracket of
sections $\zP(\za)+\za$, associated with `linear 1-forms' $\za$, in coordinates $\za=y^i(x)\xd\zx_i$. The
above theorem provides a generalization of this fact and, for each Dirac-Lie algebroid, describes the induced
Lie algebroid structure on its velocity bundle.

\bigskip
The next theorem characterizes the core bundle of a Dirac algebroid in terms of its anchor relation.
\begin{theo}\label{main0}  The core bundle $\cC_D\subset\sT^\ast M\oplus_M E^\ast$ of a Dirac algebroid $D\subset\cT E^\ast$
is the annihilator subbundle $\Ve_D^0\subset\sT^\ast M\oplus_M E^\ast$ of the anchor relation $\Ve_D\subset\sT
M\oplus_M E$.
\end{theo}
\begin{proof}
To an element $d\in D$ that projects onto $(\zt_1, \zt_2)(d)=(\zm_x,v_x)$ we can add any element $u_x$ of the
$x$-fiber of the core not changing the projections, so, due to isotropy, $\la v_x,u_x\ran=0$ for all
$v_x\in(\Ve_D)_x$ and $\cC_D\subset\Ve_D^0$. The equality follows from the conditions on the rank. In
coordinates, $d$ is represented by
$$d={\dot{x}^a}\pa_{x^a}+f_i\pa_{\zx_i}+{y^i}\xd\zx_i+g_a\xd x^a$$
and
$$u_x={\dot{\zx_i}}\pa_{\zx_i}+{p_a}\xd x^a\,.$$
Since $(d|d)=0$ and $(d+u_x|d)=0$, we have $\la \zt_2(d),u_x\ran=0$, i.e.
$${\dot{x}^a}{p_a}+{y^i}{\dot{\zx_i}}=0\,.
$$

\end{proof}

In order to describe the local form of a Dirac algebroid $D$, note first that, since an arbitrary Dirac
algebroid $D\subset \cT E^\ast$ is the restriction to the phase bundle $\Ph_D\subset E^\ast$ of a Dirac
algebroid supported on the whole bundle $E^\ast$, we can assume at the beginning for simplicity that
$\Ph_D=E^\ast$. As the Pontryagin bundle $\cT E^\ast$ is, as the bundle over the projection
$$(\zt_1,\zt_2):\cT E^\ast\ra(\sT M\op_M E)\op_M E^\ast$$
an affine bundle modeled on the pull-back bundle of the anchor bundle $\sT^\ast M\op_M E^\ast$ (Theorem
\ref{core}), we can write
\be\label{cTform}
\cT E^\ast\simeq \left(E^\ast\op_M \sT M\op_M E\right)\ti_M(\sT^\ast M\op_M E^\ast)\,.
\ee
Note that the product $\ti_M$ in the above expression is not canonical, but it can be used to express the fact
that we can add elements of $\sT^\ast_x M\op E^\ast_x$ to elements of $E^\ast_x\op \sT_x M\op E_x$ and to
serve for introducing local coordinates. Instead of the coordinates we have already used, it will be more
convenient to introduce affine coordinates $(x^a,\zh^i,\wh{\zh}^j)$in $\sT_x M\op E_x$  and dual affine
coordinates $(x^a,\zz_i,\wh{\zz}_j)$ in $\sT^\ast_x M\op E^\ast_x$, so that $(\zh^i,\wh{\zh}^j)$ represent
linear coordinates in fibers of the anchor relation $\Ve_D$ and its (non-canonical) complementary subbundle
$V$, $\sT_x M\op E_x=\Ve_D\op_M V$, respectively, and the coordinates $(\zz_i,\wh{\zz}_j)$ are linear
coordinates in the annihilators $\sT^\ast_x M\op E^\ast_x=V^0\op_M\Ve_D^0$, respectively. Note that $V^0$
represents the dual bundle $\Ve_D^\ast$.

The points of $D$ satisfy then $\wh{\zh}^j=0$. Since we can add elements of the core $\cC_D=\Ve_D^0$,
coordinates $\wh{\zz}_j$ are arbitrary. Therefore, there are sections $\wt{\zs}_i$ of $D$ associated with the
canonical local basis $\zs_i$ of sections of $\Ve_D$, $\zh^{i'}(\zs_i(x))=\zd^{i'}_i$, which read
$$\zs_i(x^a,\zx_j)=\left(x^a,\zx_j,\zh^{i'}=\zd^{i'}_i,0,\zz_{i'}=c_{ii'}^j(x)\zx_j,0\right)\,.$$
Due to isotropy, we have skew-symmetry $c_{ii'}^j(x)=-c_{i'i}^j(x)$. Now, we can add linear constraint $\Ph_D$
in $E^\ast$ by introducing affine coordinates, say $(x,\wh{x},\zx\wh{\zx})$, such that $\Ph_D$ is expressed by
$\wh{x}=0$, $\wh{\zx}=0$. In this way we get the following
\begin{theo} {\rm (local form of a Dirac algebroid)}

\noindent In the introduced local affine coordinates the Dirac algebroid $D$ consists of points
$(x,\wh{x},\zx,\wh{\zx},\zh,\wh{\zh},\zz,\wh{\zz})$ for which
\be\label{locall}
\wh{x}=0\,,\quad \wh{\zx}=0\,,\quad\wh{\zh}=0\,,\quad \zz_{k}=c_{ik}^j(x)\zh^i\zx_j\,.
\ee
Moreover, $c_{ik}^j(x)=-c_{ki}^j(x)$.
\end{theo}
Let us note that the above constraints can be viewed as a common generalizations of (\ref{kk}) and
(\ref{kkk}). The functions $c_{ik}^j$ play the role of structure functions and $\wh{\zh}=0$ defines the anchor
relation. We can write $(\zh,\wh{\zh})$ as linear functions of variables $(\dot{x},y)$  and $(\zz,\wh{\zz})$
as linear functions of $(p,\dot{\zx})$ (with coefficients being functions of $x$) to derive constraints
\be\label{localll}
\wh{\zh}(x,\dot{x},y)=0\,,\quad \zz_{i}(x,p,\dot{\zx})+c_{ik}^j(x)\zh^k(x,\dot{x},y)\zx_j=0\,.
\ee
\begin{example}\label{egen}
For the $\zP$-graph Dirac algebroid as described in example \ref{e1} we have
$$\zh=y\,,\ \wh{\zh}^b=\dot{x}^b-\zr^b_k(x)y^k\,,\ \zz_j=\dot{\zx}_j+\zr^a_j(x)p_a\,,\ \wh{\zz}_b=p_b
$$
and the equations (\ref{localll}) read
$$\dot{x}^b-\zr^b_k(x)y^k=0\,,\ \dot{\zx}_j + \zr^a_j(x) p_a+c^k_{ji}(x) y^i\zx_k=0\,,$$
exactly as in (\ref{kk}).
\end{example}

\section{Induced Dirac algebroids}
In this section we will show how appropriate linear (or affine) `nonholonomic constraints' in the velocity
bundle $\Ve_D$ give rise to new (induced) Dirac algebroids. These construction may be viewed as a
generalization of
the similar construction for the canonical Lie algebroid $E=\sT M$ in
\cite{YM1}.

Consider a Dirac algebroid $D\subset\cT E^\ast$ and let $V$ be a vector subbundle of the velocity bundle
$\Ve_D\subset\sT M\oplus_M E$ supported on $S\subset M_D\subset M$. Let $\wt{V}=(\zt_2^D)^{-1}(V)$ be the
restriction of the vector bundle $\zt_2^D:D\ra\Ve_D$ to the submanifold $V$ in the base, and let $V^0\subset
\sT^\ast M\oplus E^\ast$ be the annihilator of $V$. Of course, $V^0$ is supported on $S$ as well and
$V^0\supset \Ve_D^0=\cC_D$. Since $\sT^\ast M\oplus E^\ast$ is the core of $\cT E^\ast$, we may add vectors
from $V^0$ to vectors of the vector bundle $\zt_1:\cT E^\ast\ra E^\ast$ not changing any of two projections.
In this sense, $D^V=\wt{V}+V^0$ is again a double vector subbundle of $\cT E^\ast$ which is no longer $D$, but
still projects on $V$ {\it via} $\zt_2$, and on $\Ph_D$ {\it via} $\zt_1$.

\begin{theo}\label{tnh0} The double vector subbundle $D^V$ in $\cT E^\ast$ is a Dirac algebroid on $E$.
\end{theo}
\begin{proof} The subbundle $D^V$ is isotropic by definition, since $\wt{V}$ is isotropic as a subbundle of $D$,
and $V^0$ is isotropic and orthogonal to $\wt{V}$. The rank of this bundle is maximal, since first we loose
rank by $\dim(\Ve_D/V)$ and then we gain $\dim(V^0/\Ve_D^0)=\dim(\Ve_D/V)$.
\end{proof}

\begin{deff} The Dirac algebroid $D^V$ we will call the Dirac algebroid {\it induced from $D$ by the subbundle $V\subset \Ve_D$}.
\end{deff}

Quite similarly, we can induce affine Dirac algebroids using an affine subbundle $A$ of $\Ve_D$ based on a
submanifold $S\subset M$. Let $V=\sv(A)$ be its model vector bundle viewed as a vector subbundle of $\Ve_D$.
Let us put $\wt{A}=(\zt_2^D)^{-1}(A)$, and let $V^0\subset \sT^\ast M\oplus E^\ast$ be the annihilator of $V$.
The vector subbundle $\wt{A}$ of $\zt_2^D:D\ra\Ve_D$ is simultaneously an affine subbundle of
$\zt_1^D:D\ra\Ph_D$, thus {\it vector-affine} subbundle. Similarly as above, $D^A=\wt{A}+V^0$ is again a
vector-affine subbundle of $\cT E^\ast$ which still projects on $A$ {\it via} $\zt_2$, and on $\Ph_D$ {\it
via} $\zt_1$. Analogously to theorem \ref{tnh0} one can prove the following.

\begin{theo}\label{tnh9} The vector-affine subbundle $D^A$ in $\cT E^\ast$ is an affine Dirac algebroid on $E$.
\end{theo}
\begin{deff} The affine Dirac algebroid $D^A$ we will call the affine Dirac algebroid
{\it induced from $D$ by the affine subbundle $A\subset \Ve_D$}.
\end{deff}

\begin{example}\label{nhco} Consider  a Dirac algebroid $D_\zP$ on a vector bundle $\zt:E\ra M$, associated with a linear bivector field $\zP$ on $E^\ast$. Since the anchor relation $\Ve_{D_\zP}$ is in this case the graph of the anchor map $\zr:E\ra\sT M$, subbundles $V$ of $\Ve_{D_\zP}$ may be identified with subbundles $V_0$ of $E$, $V=\{ \zr(v)+v;v\in V_0\}$.

It is convenient to see all this in local coordinates $(x^a,\zx_i,\dot{x}^b,\dot{\zx}_j,p_c,y^k)$ in $\cT
E^\ast$. We may choose local coordinates in $(x^a)=(x^\za,x^A)$ in $M$, so that $S$ is given locally by
$x^A=0$. Let us also use linear coordinates $(y^i)$ in the fibers of $E$, so that $y=(y^i)=(y^\zi,y^I)$ and
the subbundle $V_0$ is defined by the constraint $y^I=0$. On $\cT E^\ast$ we have then local coordinates
$(x^a,\dot{x}^b,\dot{\zx}_l,p_c,y^\zi,y^I)$ where we have also decompositions $(\zx_k)=(\zx_\zk,\zx_K)$ and
$(\dot{\zx}_l)=(\dot{\zx}_\zl,\dot{\zx}_L)$ associated with the decomposition $(y^i)=(y^\zi,y^I)$. The double
subbundle $\wt{V}$ is defined by the constraints (cf. example \ref{e1})
$$\wt{V}=\{(x^a,\zx_i,\dot{x}^b,\dot{\zx}_j,p_c,y^k):\quad x^A=0\,,\ y^I=0\,,\ \dot{x}^b=\zr^b_\zi(x)y^\zi\,,\ \dot{\zx}_k=c^j_{\zi k}(x) y^\zi\zx_j-\zr^a_k(x) p_a\}\,.
$$
The points $(x^a,p_b,\dot{\zx}_i)$ of $\sT^\ast M\oplus_M E^\ast$ belong to $V^0$ if and only if $x^A=0$ and
$p_b\zr^b_\zi(x)y^\zi+\dot{\zx}_\zi y^\zi=0$ for all $(y^\zi)$, thus $\dot{\zx}_\zi=-\zr^b_\zi(x)p_b$ and
$\dot{\zx}_I$ are arbitrary. As the first condition agrees with the original constraints, we get the final
constraints defining $D_\zP^V$:
\be\label{con}
x^A=0\,,\ \dot{x}^b=\zr^b_\zi(x)y^\zi\,,\ \dot{\zx}_\zk=c^j_{\zi \zk}(x) y^\zi\zx_j-\zr^a_\zk(x) p_a\,,\
y^I=0\,,
\ee
as adding ${V}^0$ makes $\dot{\zx}_K$ arbitrary.

Let us assume now that $\zP$ is a Poisson tensor, i.e., $D_\zP$ is a Dirac-Lie algebroid. The first
integrability condition for $D_\zP^V$ is now $\pr_{\sT E^\ast}\subset \sT E^\ast_{|S}$, i.e.,
$\dot{x}^B=\zr^B_\zi(x^\za,0)y^\zi=0$ for all $y^\zi$, thus
\be\label{dd1}\zr^B_\zi(x^\za,0)=0\,.\ee
To check the second integrability condition, let us note first that $D_\zP^V$ is locally generated by the
sections $\pa_{\zx_I}$ and  $R_{\zx_\zi}=\wt{\zP}(\xd\zx_\zi)+\xd\zx_\zi$. Since, by the assumption that $\zP$
is Poisson,
$$\lna R_{\zx_\zi},R_{\zx_{\zi'}}\rna=R_{\zP(\xd\zx_\zi,\xd\zx_{\zi'})}\,,$$
$I$th components of $\xd\left(\zP(\xd\zx_\zi,\xd\zx_\zi')\right)$ must vanish along $E^\ast_{|S}$, i.e.,
\be\label{dd2}c^I_{\zi\zi'}(x^\za,0)=0\,.\ee
The vector fields $\pa_{\zx_I}$ commute, so it remains to check whether $\lna\pa_{\zx_I},R_{\zx_\zi}\rna$ are
sections of $D_\zP^V$ over $E^\ast_{|S}$. But
\be\label{ddd}\lna\pa_{\zx_I},R_{\zx_\zi}\rna=[\pa_{\zx_I},\wt{\zP}(\xd\zx_\zi)]\ee
and, according to (\ref{dd1}),
$$\wt{\zP}(\xd\zx_\zi)(x^\za,0,\zx_i)=c_{\zi\zi'}^{\zi''}(x^\za,0)\zx_{\zi''}\pa_{\zx_\zi'}+f^{I'}(x^\za,0,\zx_i)\pa_{\zx_{I'}}$$
for some functions $f^{I'}$. Hence
$$\lna\pa_{\zx_I},R_{\zx_\zi}\rna(x^\za,0,\zx_i)=\pa_{\zx_I}(f^{I'})(x^\za,0,\zx_i)\pa_{\zx_{I'}}\,,$$
so that the expression in (\ref{ddd}) is spanned over $E^\ast_{|S}$ by $\pa_{\zx_I}$. Thus we get that
$D_\zP^V$ is a Dirac-Lie algebroid if and only if (\ref{dd1}) and (\ref{dd2}) are satisfied. This, in turn,
means that $V_0$ is a Lie subalgebroid in the Lie algebroid on $E$ associated with $\zP$, so $V$ is a Lie
subalgebroid of the Lie algebroid $\Ve_D$ -- the graph of the anchor map.
\end{example}
\begin{theo} If $D_\zP$ is a $\zP$-graph Dirac-Lie algebroid, then $D_\zP^V$ is a Dirac-Lie algebroid if
and only if $V$ is a Lie subalgebroid of $\,\Ve_D$.
\end{theo}
\begin{example}\label{nhco1} A particular case of the above example is the canonical Dirac-Lie algebroid $D_M$.
In this case we recover the induced Dirac structure considered in \cite{YM1}, i.e., the set
\begin{multline*}
D_V=\{\,(X+\za)\in\sT\sT^\ast M\oplus_{\sT^\ast M}\sT^\ast\sT^\ast M; \\
X\in(\sT\zp_M)^{-1}(V_0),\quad\forall\, W\in (\sT\zp_M)^{-1}(V_0):\quad\langle \za, X\rangle=\zw_M(X,W)\;\},
\end{multline*}
where $\zw_M$ is the canonical symplectic form on $\sT^\ast M$. Let us show that this indeed is the case.
\medskip

According to our definition the canonical Dirac-Lie algebroid $D_M$ on the cotangent bundle is given by the
canonical Poisson structure $\zP_M$ or the canonical symplectic structure $\zw_M$ on $\sT^\ast M$, i.e.,
$$D_M=\Graph(\zP_M)=\Graph(\zw_M).$$
The velocity bundle $\Ve_{D_M}\subset \sT M\oplus_M\sT M$ is in this case the graph of the identity map on
$\sT M$, the phase bundle $\Ph_{D_M}$ is the whole cotangent bundle $\sT^\ast M$ and the core
$\cC_{D_M}\subset\sT^\ast M\oplus_M\sT^\ast M$ is the graph of the minus identity map on $\sT^\ast M$.

Any subbundle $V$ of the velocity bundle is given by a subbundle $V_0$ of the tangent bundle $\sT M$ and is of
the form $V=\{v+v;\; v\in V_0\}$. Then we get
$$\widetilde{V}=\left(\zt^{D_M}_2\right)^{-1}(V)=\{\,X+\tilde\zw_M(X)\in\cT\sT^\ast M:\quad \sT\zp_M(X)\in V_0\,\}.$$

The anihilator $V^0$ consists of all pairs of covectors $(\zf,\zc)$ at the same point in $M$ such that
$\zf+\zc\in (V_0)^0$. Since the induced Dirac structure is $D^V_M=\tilde V+V^0$, we have that
$$D^V_M=\{(X+\zf)+(\tilde\zw_M(X)+\zc);\quad \sT\zp_M(X)\in V_0,\; \zf+\zc\in (V_0)^0\,\}.$$
The "+" sign in brackets in the above formula stands for adding an element of a core to an element of double
vector bundle. To compare $D^V_M$ with the Dirac structure considered in \cite{YM1} let us observe, that the
projection of $D^V_M$ on $\sT\sT^\ast M$ gives the whole $(\sT\zp_M)^{-1}(V_0)$. Adding elements of a core of
a double vector bundle does not change projections, therefore adding $\zf$ to $X$ produces another element $Y$
of $(\sT\zp_M)^{-1}(V_0)$. Since $\tilde\zw_M$ is a double vector bundle isomorphism, it respects the
structure of the double vector bundle. In particular, it maps the core of $\sT\sT^\ast M$ to the core of
$\sT^\ast\sT^\ast M$. Both cores are isomorphic to $\sT^\ast M$, and  $\tilde\zw_M$ restricted to the core is
the identity map. We have
$$\tilde\zw_M(Y)+\zc=\tilde\zw_M(X+\zf)+\zc=\tilde\zw_M(X)+(\zf+\zc)\,,$$
so
$$D^V_M=\{\,X+(\tilde\zw_M(X)+\zc);\quad \sT\zp_M(X)\in V_0,\;\zc\in (V_0)^0\;\}\,.$$
Evaluating $\tilde\zw_M(X)+\zc$ on any $W\in(\sT\zp_M)^{-1}(V_0)$, we get that
$$\langle\,\tilde\zw_M(X)+\zc, W\,\rangle=
\langle\,\tilde\zw_M(X), W\,\rangle+\langle\, \zc,\sT\zp_M(W)\,\rangle= \langle\,\tilde\zw_M(X),
W\,\rangle=\zw_M(X,W)\,,$$ thus $D^V_M\subset D_V$. For dimensional reasons the inclusion is in fact equality.

\end{example}

\section{Lagrangian and Hamiltonian formalisms based on Dirac \\ algebroids}
\subsection{Implicit differential equations}
Let us start with an explanation what we will understand as implicit dynamics on a manifold $N$.
\begin{deff} An {\it ordinary first-order (implicit) differential equation (implicit dynamics)} on a manifold $N$
will be understood as a subset $\cD$ of the tangent bundle $\sT N$. We say that a smooth curve $\zg:\R\ra N$
(or a smooth path $\zg:[t_0,t_1]\ra N$) {\it satisfies the equation} $\cD$ (or {\it is a solution of $\cD$}),
if its tangent prolongation $\dot{\zg}:\R\ra\sT N$ (resp., $\dot{\zg}:[t_0,t_1]\ra\sT N$) takes values in
$\cD$. A curve (or a path) $\wt{\zg}$ in $\sT N$ we call {\it admissible}, if it is the tangent prolongation
of its projection ${\wt{\zg}}_N$ on $N$.
\end{deff}
According to the above definition, solutions of an implicit dynamics $\cD$ on a manifold $N$ are projections
$\wt{\zg}_N$ of admissible curves $\wt{\zg}$ lying in $D$. Note, however, that different implicit differential
equations may have the same set of solutions. First of all, if $\cD$ is supported on a subset $N_0$,
$\zt_N(\cD)=N_0$, only vectors from $\cD\cap\sT N_0$ do matter, if solutions are concerned. Hence,
$\cD'=\cD\cap\sT N_0$ has the same solutions as $\cD$, and $\cD\subset \sT N_0$ is the {\it first
integrability condition}. Of course, replacing $\cD$ with $\cD'$ may turn out to be an infinite procedure, but
we will not discuss the integrability problems in this paper.

All this can be generalized to ordinary implicit differential equations of arbitrary order. In this case we
consider $\cD$ as a subset of higher jet bundles, the $n$-th tangent bundle $\sT^n N$ in case of an equation
of order $n$, and consider $\zg$ as a solution when its $n$-th jet prolongation takes values in $\cD$. If we
call the $n$-th jet prolongations {\it admissible} in $\sT^n N$, then solutions of $\cD$ are exactly
projections $\wt{\zg}_N$ to $N$ of admissible curves (or paths) $\wt{\zg}$ in $\sT^n N$ lying in $\cD$.

\begin{rem} The implicit differential equations described above are called by some authors {\it differential relations}.
Let us explain that we use the most general definition, not requiring from $\cD$ any differentiability
properties, since in real life the dynamics $\cD$ we encounter are often not submanifolds. This generality is
also very convenient, as allows us to skip technical difficulties in the corresponding Lagrangian and
Hamiltonian formalisms. Of course, what is a balast in defining implicit dynamics can be very useful in
solving the equations, but in our opinion, solving could be considered case by case, while geometric
formalisms of generating dynamics should be as general as possible. Note also that for any subset $N_0$ of a
manifold $N$ the tangent prolongations $\sT N_0$, $\sT^2 N_0$, etc., make precisely sense as subsets of $\sT
N$, $\sT^2 N$, etc. They are simply understood as families of the corresponding jets of appropriately smooth
curves in $N$ which take values in $N_0$.
\end{rem}
Admissibility of a path in $\sT N$ has a natural generalization for paths $\zg$ in an algebroid $E$. This
concept plays a fundamental role in the `integration' of Lie algebroids to Lie groupoids \cite{CF} and appears
as natural consequence of the algebroid version of the Euler-Lagrange equations \cite{GGU3,GG}. We propose the
following extension of this concept to Dirac algebroids, which reduces to the standard definition for
$\zP$-graph Dirac algebroids and Lie algebroids.

Note first that given a smooth curve or path $\zg$ with values in $E$ we have a unique `tangent prolongation'
of $\zg$ to a curve (or path) $\bzg:\R\ra \sT M\oplus_M E$ (resp., $\bzg:[t_0,t_1]\ra \sT M\oplus_M E$),
defined in obvious way by
\be\label{aprolong} \bzg(t)=\dot{\zg_M}(t)\oplus\zg(t)\,,
\ee
where ${\zg_M}$ is the projection of $\zg$ to $M$, ${\zg_M}=\zt\circ\zg$.
\begin{deff} Let $D$ be a Dirac algebroid on $\zt:E\ra M$ and let $\Ve_D\subset\sT M\oplus_M E$ be its anchor relation.
We say that a curve $\zg:\R\ra E$ (or a path $\zg:[t_0,t_1]\ra E$) is {\it $D$-admissible}, if its tangent
prolongation $\bzg$ takes values in $\Ve_D$,
$$\forall\ t\in\R\ [\bzg(t)=\dot{\zg_M}(t)+\zg(t)\in \Ve_D\subset \sT M\oplus_M E]\,.$$
\end{deff}
\begin{rem}\label{admiss}It is easy to see that in the case of a $\zP$-graph Dirac algebroid, when $\Ve_D$ is
the graph of the anchor map $\zr:E\ra\sT M$, a curve $\zg$ in $E$ is $D$-admissible if and only if
$\zr(\zg(t))=\dot{\zg_M}(t)$ that coincides with the concept of admissibility for Lie algebroids. In
particular, for the canonical Lie algebroid $E=\sT M$ and the corresponding canonical Dirac algebroid $D_M$, a
curve $\zg$ in $\sT M$ is $D_M$-admissible if and only if it is admissible, i.e., it is the tangent
prolongation of its projection $\zg_M$ on $M$, $\zg(t)=\dot{\zg_M}(t)$.
\end{rem}

\subsection{Phase dynamics, Hamilton, and Euler-Lagrange equations}
Our experience in working with (constrained) systems on skew-algebroids \cite{GGU3,GG} suggests us the
following approach. Let us fix a Dirac algebroid $D$ on a vector bundle $E$,
$$D\subset \sT^\ast E^\ast\oplus_{E^\ast}\sT E^\ast\simeq\sT^\ast E\oplus_{E^\ast}\sT E^\ast\,.
$$
In generalized Lagrangian and Hamiltonian formalisms we will view $D$ as a differential relation
$$\ze_D:\sT^\ast E\rel\sT E^\ast$$
or
$$\zb_D:\sT^\ast E^\ast\rel\sT E^\ast\,,$$
respectively. We use the symbol `$\rel$' to stress that we deal with relations having domains in $\sT^\ast E$
or $\sT^\ast E^\ast$ (not necessarily the whole $\sT^\ast E$ or $\sT^\ast E^\ast$) and with ranges being
subsets of $\sT E^\ast$. Note that $\ze_D=\zb_D\circ\cR_\zt$ and $\zb_D$ is a relation over the identity on
the support of $D$ in $E^\ast$ -- the phase bundle $\Ph_D$. The bundle $E$ plays the role of the bundle of
generalized velocities (quasi-velocities), and its dual, $E^\ast$, the role of the phase space.

A Lagrangian function $L:E\ra\R$ and a Hamiltonian $H:E^\ast\ra\R$ give rise to maps associated with their
derivatives, $\xd L:E\ra\sT^\ast E$ and $\xd H:E^\ast\ra\sT^\ast E^\ast$, respectively. The Lagrangian
produces the {\it phase dynamics} $\ze_D[\xd L]$ as the image of $E$ under the composition of relations
$\zL_D^L=\ze_D\circ\xd L$:
\be\label{Ld}\ze_D[\xd L]=\zL_D^L(E)\subset\sT E^\ast\,.
\ee
The relation $\zL_D^L$ we call the {\it Tulczyjew differential of $L$}. Similarly, when using the composition
of relations $\zF_D^H=\zb_D\circ\xd H$, that projects onto the relation $\zq_D^H=\zt_{E^\ast}\circ\zF_D^H$
being the identity on a subset of $E^\ast$, the {\it Hamiltonian dynamics} generated by the Hamiltonian $H$ is
defined by
\be\label{Hd}\zb_D[\xd H]=\zF_D^H(E^\ast)\subset\sT E^\ast\,.
\ee
The phase dynamics $\ze_D[\xd L]$ associated with the Lagrangian $L$ {\it has a Hamiltonian description}, if
there is a hamiltonian $H$ with the same dynamics, $\ze_D[\xd L]=\zb_D[\xd H]$.

Of course, the actual phase spaces associated with $L$ and $H$ are projections of the phase dynamics on
$E^\ast$, $\Ph_D^L=\zt_{E^\ast}(\ze_D[\xd L])$ and $\Ph_D^H=\zt_{E^\ast}(\zb_D[\xd H])$.

Since, as easily seen, the projection of the relation $\zL_D^L=\ze_D\circ\xd L$ to $E\oplus_ME^\ast$ is
actually a function,
$$\zl_D^L=\zt_{E^\ast}\circ\zL_D^L=\sT^\ast\zt\circ\xd L_{|\Ve_D^L}\,,$$
called the {\it Legendre map} associated with the Lagrangian $L$. The domain of the Legendre map will be
denoted $\Ve_D^L$ and called the {\it Euler-Lagrange domain}. It is easy to see that $\Ph_D^L$ is the image of
the  Legendre map
$${\zl_D^L}=\zt_{E^\ast}\circ\zL_D^L: E\supset \Ve_D^L\ra \Ph_D^L\subset E^\ast\,,$$
and the Legendre map is the restriction of the vertical derivative $\xd^\sv L:E\ra E^\ast$ to
\be\label{VeL}\Ve_D^L=\{ v\in \pr_E(\Ve_D): \xd^\sv L(v)\in\Ph_D\}\,.
\ee
In local coordinates,
$$\xd^\sv L(x,y)=\left(x,\frac{\pa L}{\pa y^i}(x,y)\right)\,.$$
If $D$ is a $\zP$-graph Dirac algebroid, then $\Ve_D^L=E$.

The diagram picture for the corresponding {\it Tulczyjew triple} containing: $\sT^\ast E$ (the Lagrangian
side), the canonically isomorphic (via $\cR_\zt)$ double vector bundle $\sT^\ast E^\ast$ (the Hamiltonian
side), and $\sT E^\ast$ (the phase dynamics side) is the following (here, the arrows denote relations):

\be\label{diag}\xymatrix @+1pc{
\sT^\ast E\ar@<1ex>[d]^{\zp_{E}}\ar[rr]^{\ze_D}\ar@/^1.5pc/[rrrr]^{\cR_\zt}  && \sT
E^\ast\ar[d]^{\zt_{E^\ast}}&&
\sT^\ast E^\ast\ar[ll]_{\zb_D}\ar@<1ex>[d]^{\zp_{E^\ast}}\\
E\ar[rr]^{{\zl_D^L}}\ar@<1ex>[u]^{\xd L}\ar@{.>}[rru]^{\Lambda_{D}^L} && E^\ast && E^\ast
\ar[ll]_{\zq_D^H}\ar@<1ex>[u]^{\xd H}\ar@{.>}[llu]_{\zF_{D}^H}}.
\ee

\bigskip
The Euler-Lagrange equation associated with $L$ will be viewed as an implicit dynamics on $E$. It will make
sense for curves in $E$ taking values in the Euler-Lagrange domain $\Ve_D^L\subset E$.
\begin{deff}
We say that a curve $\zg:\R\ra\Ve_D^L$ {\it satisfies ({\rm or} is a solution of) the Euler-Lagrange
equation}, if $\zg$ is $\zL_D^L$-related to an admissible curve $\wt{\zg}$ in $\sT E^\ast$ (i.e., $\wt{\zg}$
is the tangent prolongation of its projection ${\wt{\zg}}_{E^\ast}$ onto $E^\ast$). In particular, ${\zg}$ is
${\zl_D^L}$-related to the curve ${\wt{\zg}}_{E^\ast}$ which satisfies the phase equation.
\end{deff}
To describe the Euler-Lagrange equation explicitly, consider the tangent prolongation of the relation
$\zL_D^L$,
$$\sT\zL_D^L=\sT\ze_D\circ\sT\xd L:\sT E\rel\sT\sT E^\ast\,.$$

In $\sT\sT E^\ast$ we can distinguish {\it holonomic vectors}, i.e., vectors $X_v\in \sT_v\sT E^\ast$ such
that $v$ equals the tangent projection of $X_v$ onto $\sT E^\ast$, i.e., $v=\sT\zt_{E^\ast}(X_v)$. The set of
holonomic vectors can be seen as the second tangent bundle $\sT^2E^\ast$. We define the {\it (implicit)
Euler-Lagrange dynamics} as the subset of $\sT E$ defined by the inverse image
$$\el_D^L=(\sT\zL_D^L)^{-1}(\sT^2 E^\ast)\subset\sT E\,.$$
\begin{theo} If a curve $\zg:\R\ra E$ satisfies the Euler-Lagrange equation, then its tangent prolongation
takes values in $\el_D^L$. In particular, $\zg$ is {\it $D$-admissible}.
\end{theo}
\begin{proof}\label{eld} Let $\wt{\zg}$ be an admissible curve, $\wt{\zg}=\dot{\wt{\zg}}_{E^\ast}$,
contained in $\ze_D[\xd L]$ and $\zL_D^L$-related to $\zg$. Then, its tangent prolongation $\dot{\wt{\zg}}$ is
$\sT\zL_D^L$-related to the tangent prolongation $\dot{\zg}$ of $\zg$. But $\dot{\wt{\zg}}$ is the 2-tangent
prolongation of $\wt{\zg}_{E^\ast}$, thus lies in $\sT^2E^\ast$.
\end{proof}

Note that the converse is `almost true'. Indeed, if $\dot{\zg}$ lies in $\sT(\zL_D^L)^{-1}(\sT^2 E^\ast)$, we
only need to know that we can pick up a curve in $\sT^2 E^\ast$ being $\sT\zL_D^L$-related to $\dot{\zg}$.
This can be assured, for instance, by some smooth transversality assumptions. As we do not want to consider
these questions here, let us only mention that the converse of theorem \ref{eld} is always true in the case
when $\zL_D^L$ is a map, for instance for $\zP$-graph Dirac algebroids.
\begin{rem}
Let us observe that in our setting the Euler-Lagrange equation is a first-order equation on $E$, in full
agreement with the fact that the Hamilton equation is first-order as well. In the standard setting, the
Euler-Lagrange equation is viewed as second-order, but for curves in the base $M$. This can be explained as
follows. The solutions of the Euler-Lagrange equations are always $D$-admissible. In the case of the canonical
algebroid $E=\sT M$ the admissible curves in $\sT M$ are exactly the tangent prolongations of curves in the
base $M$, thus we may view the corresponding Euler-Lagrange equations as first-order equations on tangent
prolongations, so second-order equations for curves on the base.
\end{rem}
\subsection{Hyperregular Lagrangians}
Let us assume that we have a {\it hyperregular} Lagrangian $L:E\ra\R$, i.e., such a Lagrangian that its
vertical derivative $\mathfrak{L}=\xd^\sv L:E\ra E^\ast$ is a diffeomorphism. For instance, $L$ can be  {\it
of mechanical type}, being the sum of a `kinetic energy' (associated with a `metric' on the vector bundle $E$)
and a potential (a basic function). It is well known \cite{GGU3} that in this case the Hamiltonian
$H:E^\ast\ra\R$ defined by
\be\label{rl}H=(\nabla_E(L)-L)\circ\mathfrak{L}^{-1}\,,\ee
where $\nabla_E$ is the Euler vector field on the vector bundle $E$, defines the same Lagrangian submanifold
in $\sT^\ast E^\ast$ as $L$ in $\sT^\ast E$, when we identify canonically both bundles:
$$\xd H(E^\ast)=\cR_\zt(\xd L(E))\,.$$
In local coordinates, $\zx_i=\frac{\pa L}{\pa y_i}(x,y)$ and
$$H(x,\zx)=\zx_i\cdot y^i(x,\zx)-L(x,y(x,\zx))\,.$$
It is then easy to see that the Legendre map ${\zl_D^L}$ is a diffeomorphism of $\Ve_D^L$ on $\Ph_D^L$, and
that the phase dynamics associated with $L$ and $H$ coincide.
\begin{theo} If $L$ is a hyperregular Lagrangian, then,  for any Dirac algebroid $D$, the phase dynamics $\ze_D[\xd L]$ coincides with the phase dynamics $\zb_D[\xd H]$ for the Hamiltonian $H$ defined by (\ref{rl}). In this sense,
for hyperregular Lagrangians, the Lagrangian and Hamiltonian formalisms are equivalent.
\end{theo}

\subsection{Constraints}

{\it Nonholonomic linear (or affine) constraints} in our Dirac algebroid setting are understood as represented
by vector (affine) subbundles $V$ of the the velocity bundle $\Ve_D$. This could look strange for the first
sight, but it becomes quite natural, if we recall that the solution of the Euler-Lagrange equations are
admissible curves $\zg$ in the bundle $E$ of quasi-velocities. Since there is a canonical tangent prolongation
$\bzg$ of $\zg$, with $\bzg$ lying in $\Ve_D$, the constraint $V$ gives us equations for $\zg$ with $\bzg$ in
$V$. The general principle is the following.
\begin{deff} (Nonholonomic constraints) The phase dynamics and the Euler-Lagrange equations for a
constrained Lagrangian system on a Dirac algebroid $D$ over a vector bundle $\zt:E\ra M$, and associated with
the Lagrangian $L:E\ra \R$ and the linear (affine) constrain bundle $V\subset\Ve_D$, is the dynamics
associated with the same Lagrangian but on the induced Dirac algebroid $D^V$ over $E$.
\end{deff}

Another type of constraints we can consider in our setting are {\it vakonomic constraints} represented by a
submanifold (not necessary an affine subbundle) $C$ of $E$. Let us recall that with any submanifold $C$ of $E$
and any function $L:C\ra\R$ we can associate a lagrangian submanifold $[\xd L_C]$ of $\sT^\ast E$ defined by
\be\label{lasu} [\xd L_C]=\{ \zh\in\sT^\ast_y E:y\in C\ \text{and}\ \forall v\in\sT_yC\quad\la\zh,v\ran=\la\xd
L(y),v\ran\}\,.
\ee
We can view $[\xd L_C]$ as a relation $[\xd L_C]:E\rel\sT^\ast E$. Now we can define the constrained phase
dynamics and the Euler-Lagrange equations completely analogously to unconstrained ones, but replacing the
relation $\xd L(E)$ with $[\xd L_C]$.
\begin{deff} (Vakonomic constraints) The phase dynamics for a
constrained Lagrangian system on a Dirac algebroid $D$ over a vector bundle $\zt:E\ra M$, associated with the
Lagrangian $L:E\ra \R$ and a vakonomic constraints represented by a submanifold $C$ of $E$, is the dynamics
represented by the subset $\ze_D([\xd L_C])$ of $\sT E^\ast$. We say that a curve $\zg:\R\ra[\xd L_C]$ {\it
satisfies the vakonomically constrained Euler-Lagrange equation}, if $\zg$ is $\ze_D$-related to an admissible
curve in $\sT E^\ast$.
\end{deff}
\begin{rem} Note first that, by definition, the phase dynamics for vakonomic constraints depends on
the restriction of the Lagrangian $L$ to $C$ only. Second, we recover the old dynamics in the unconstrained
case, as $\xd L(E)=[\xd L_{E}]$. This can look strange at the first sight that we define solution of the
Euler-Lagrange equations as curves in $[\xd L_C]$ and not in $\Ve_D\subset E$, but when the constraints are
absent there is no real difference between $\Ve_D$ and $[\xd L_{\Ve_D}]$, since the projection $\zp_E$
establishes a diffeomorphism. In the presence of a constraint we have no longer this diffeomorphism. Of
course, we could say that a curve in $\Ve_D$ satisfies the constrained Euler-Lagrange equation, if it is a
projection of an appropriate curve in $[\xd L_C]$, but our approach seems to be more natural. It could happen
that one curve is the projection of different curves in $[\xd L_C]$ that is a geometric interpretation of the
presence of `Lagrange multipliers'.
\end{rem}

\section{Examples}
\begin{example}\label{e5.5} {\bf (Mechanics on a general Dirac algebroid)}

\medskip\noindent
The very general scheme of the phase or the Euler-Lagrange dynamics on a Dirac algebroid $D\subset\cT E^\ast$
can be described in local coordinates as follows. Let us choose the standard adapted coordinates (slightly
reordered) $(x,\zx,\dot{x},y,p,\dot{\zx})$ in $\cT E^\ast$. Starting with a Lagrangian $L:E\ra\R$ we can
define the associated subset $[[\xd L]]$ in $\cT E^\ast$ as consisting of points with coordinates for which
$\zx=\frac{\pa L}{\pa y}(x,y)$
and $p=-\frac{\pa L}{\pa x}(x,y)$. Next, we intersect $[[\xd L]]$ with $D$ getting the (implicit)
Euler-Lagrange equations defined by the following relations (in coordinates of (\ref{localll})):
\bea\label{crd}
&&\left(x,\frac{\pa L}{\pa y}(x,y)\right)\in\Ph_D\,,\quad \wh{\zh}(x,\dot{x},y)=0\,,\\
&&\quad \zz_{i}\left(x,-\frac{\pa L}{\pa x}(x,y),\frac{\xd }{\xd t}\left(\frac{\pa L}{\pa
y}(x,y)\right)\right)+c_{ik}^j(x)\zh^k(x,\dot{x},y)\frac{\pa L}{\pa y^j}(x,y)=0\,.\label{crd1}
\eea
Similarly, starting with a Hamiltonian $H:E^\ast\ra\R$ and defining the subset $[[\xd H]]$ by putting the
constraints $y=\frac{\pa H}{\pa \zx}(x,\zx)$ and $p=\frac{\pa H}{\pa x}(x,\zx)$, we get after intersecting
with $D$ the following (implicit) phase dynamics
\bea\label{crdH}
&&(x,\zx)\in\Ph_D\,,\quad \wh{\zh}\left(x,\dot{x},\frac{\pa H}{\pa \zx}(x,\zx)\right)=0\,,\\
&&\quad \zz_{i}\left(x,\frac{\pa H}{\pa x}(x,\zx),\dot{\zx}\right)+c_{ik}^j(x)\zh^k\left(x,\dot{x},\frac{\pa
H}{\pa \zx}(x,\zx)\right)\zx_j=0\,.\label{crdH1}
\eea
For the canonical Dirac algebroid $D_M$ we have in adapted coordinates $\wh{\zh}^a=\dot{x}^a-y^a$,
$\zz_a=\dot{\zx}_a+p_a$, and $c^k_{ij}=0$, so we get the standard Euler-Lagrange
$$\frac{\xd x^a}{\xd t}=y^a, \quad
\frac{\xd}{\xd t}\left(\frac{\partial L}{\partial y^a}\right)(x,y)=\frac{\partial L}{\partial x^a}(x,y)$$ and
Hamilton
$$\frac{\xd {\zx}_a}{\xd t}=- \frac{\pa H}{\partial {x^a}}(x,\zx)\,,\quad
\frac{\xd x^b}{\xd t}=\frac{\pa H}{\partial {\zx_b}}(x,\zx)
$$
equations. Changing the symbols $y,\zx$ for velocities and momenta into the standard ones, $\dot{x},p$, we end
up with the traditional Euler-Lagrange and Hamilton equations.
\end{example}

\begin{example} {\bf (Pontryagin Maximum Principle for general Dirac algebroids)}

\medskip\noindent
Starting with a general Dirac algebroid as above, let us impose a vakonomic constraint $C\subset E$
parametrized by $f:M\ti U\ra C$, with $U$ being a manifold of `control parameters'. In local coordinates
$(x,y)$ in $E$ and $u$ in $U$, the parametrization yields $y=f(x,u)$. Note that, classically, for $E=\sT M$
and $y=\dot{x}$, the constraint $C$ represents the differential equation $\dot{x}=f(x,u)$.

A Lagrangian $L:C\ra \R$ may be now seen as a function $L:M\ti U\ra\R$, and $[\xd L_C]$ consists of points
$(x,y,p,\zx)\in\sT^\ast E$  (we skip the indices) such that
\be\label{constr}
y=f(x,u)\,,\quad p=\frac{\pa L}{\pa x}-\zx\frac{\pa f}{\pa x}\,,\quad \zx\frac{\pa f}{\pa u}=\frac{\pa L}{\pa
u}\,.
\ee
The above identities define a subset $[[\xd L_C]]$ in $\cT E^\ast$ which, similarly as above, leads to
implicit Euler-Lagrange equations
\bea\label{crd-cstr}
&&(x,\zx)\in\Ph_D\,,\quad \wh{\zh}(x,\dot{x},f(x,u))=0\,,\\
&&\quad \zz_{i}\left(x,\zx\frac{\pa f}{\pa x}-\frac{\pa L}{\pa
x}(x,y),\dot{\zx}\right)+c_{ik}^j(x)\zh^k(x,\dot{x},f(x,u))\zx=0\,,\label{crd1-cstr}
\eea
constrained additionally by
\be\label{pmp}
\zx\frac{\pa f}{\pa u}-\frac{\pa L}{\pa u}=0\,.
\ee
Let us note that equations (\ref{crd-cstr}) and (\ref{crd1-cstr}) are the same as the Hamilton equations
(\ref{crdH}) and (\ref{crdH1}) with the Hamiltonian
\be\label{pmp-ham}
H(x,u,\zx)=\zx\cdot f(x,u)-L(x,u)
\ee
depending on the parameter $u$. Moreover, the equation (\ref{pmp}) reads $\frac{\pa H}{\pa u}(x,u,\zx)=0$ that
is an infinitesimal form of the Pontryagin Maximum Principle (PMP): our solutions choose control parameters
which are critical for the Hamiltonian. The whole picture is an obvious generalization of (PMP), this time for
Dirac algebroids, of course in its smooth and infinitesimal version.
\end{example}

\begin{example}\label{e6} {\bf (Mechanics on skew algebroids)}

\medskip\noindent
Consider the Dirac algebroid $D_\zP$ associated with a linear bivector field $\zP$ on $E^\ast$, as described
in example \ref{e1}. Since in this case $D_\zP$ is the graph of the map $\wt{\zP}$, the relation
$\ze_{D_\zP}$ is a map. Hence, $\zL_{D_\zP}^L=\ze_{D_\zP}\circ\xd L$ is also a map
$\zL_{D_\zP}^L:E\ra\sT E^\ast$ which in local coordinates reads
\be\zL_{D_\zP}^L(x^a,y^i)= \left(x^a,\frac{\partial L}{\partial y^i}(x,y),
\zr^b_k(x)y^k, c^k_{ij}(x) y^i\frac{\partial L}{\partial y^k}(x,y) + \zr^a_j(x)\frac{\partial L}{\partial
x^a}(x,y)\right)\,. \label{F1.4a}
\ee
The Legendre relation $\zl_{D_\zP}^L$ is also a map which reads
\be\label{leg1}\zl_{D_\zP}^L(x^a,y^i)= \left(x^a,\frac{\partial L}{\partial y^i}(x,y)\right)\,.
\ee
Let $\zg(t)=(x(t),y(t))$ be a smooth curve in $E$. Since $\wt{\zg}=\zL_{D_\zP}^L\circ\zg$ is the only curve in
$\sT E^\ast$ which is $\zL_{D_\zP}^L$-related to $\zg$, the latter satisfies the Euler-Lagrange equation if and
only if $\wt{\zg}$ is admissible, i.e., $\wt{\zg}=\dot{\ul{\wt{\zg}}}$. In local coordinates,
\be\label{EL}\qquad\frac{\xd x^a}{\xd t}(x)=\zr^a_k(x)y^k, \quad
\frac{\xd}{\xd t}\left(\frac{\partial L}{\partial y^j}\right)(x,y)= c^k_{ij}(x) y^i\frac{\partial L}{\partial
y^k}(x,y) + \zr^a_j(x)\frac{\partial L}{\partial x^a}(x,y)\,,
\ee
in the full agreement with the Euler-Lagrange equation for Lie (and general skew) algebroids as described in
\cite{GGU3,GG,LMM,Mar1,We}. Note that we do not assume any regularity of the Lagrangian.

As for the Hamilton equations, let us note that also in this case the relation $\zb_{D_\zP}$ is a map,
$\zb_{D_\zP}=\wt{\zP}$,
\be\label{bha}
\zP(x^a,\zx_j,p_b,y^i) = (x^a, \zx_j, \zr^b_k(x)y^k, c^k_{ij}(x) y^j\zx_k - \zr^a_j(x) p_a)\,.
\ee
The corresponding phase dynamics is explicit and associated with the Hamiltonian vector field
\be\label{ham} \X_H(x,\zx)=\left(c^k_{ij}(x)\zx_k\frac{\pa H}
{\partial {\zx_i}}(x,\zx)- \zr^a_j(x)\frac{\pa H}{\partial {x^a}}(x,\zx)\right) \partial _{\zx_j} +
\zr^b_i(x)\frac{\pa H}{\partial {\zx_i}}(x,\zx)\partial _{x^b}\,,
\ee
i.e.,
$$\dot{\zx_j}=\left(c^k_{ij}(x)\zx_k\frac{\pa H}
{\partial {\zx_i}}(x,\zx)- \zr^a_j(x)\frac{\pa H}{\partial {x^a}}(x,\zx)\right)\,\quad
\dot{x}^b=\zr^b_i(x)\frac{\pa H}{\partial {\zx_i}}(x,\zx)\,.
$$

In the particular case of the canonical Lie algebroid $E=\sT M$, we can take for coordinates $y$ in the fiber
the coordinates $\dot{x}^a$ induced from the base. As now $c^a_{bc}=0$ (coordinate vector fields commute) and
$\zr^a_b=\zd^a_b$ (the anchor map is the identity), we get the traditional Euler Lagrange equations
$$\frac{\xd x^a}{\xd t}=\dot{x}^a, \quad
\frac{\xd}{\xd t}\left(\frac{\partial L}{\partial \dot{x}^a}\right)(x,\dot{x})=\frac{\partial L}{\partial
x^a}(x,\dot{x})\,,$$ as a particular case. Also The Hamilton equations become completely traditional in
coordinates $\zx$ replaced by the corresponding momenta:
$$\dot{p_a}=- \frac{\pa H}{\partial {x^a}}(x,p)\,,\quad
\dot{x}^b=\frac{\pa H}{\partial {p_b}}(x,p)\,.
$$
\end{example}

\begin{example}\label{e7} {\bf (Mechanics on presymplectic manifolds)}

\medskip\noindent
Consider the Dirac algebroid $D_\zw$ associated with a linear 2-form $\zw$ on $E^\ast$, as described in
example \ref{e2}. Since in this case $D_\zw$ is the graph of the map $\wt{\zw}:\sT E^\ast\ra\sT^\ast
E^\ast\simeq\sT^\ast E$, the implicit phase dynamics associated with a Lagrangian and a Hamiltonian are
inverse images of the images of $\xd L$ and $\xd H$, respectively. In coordinates,
$$\zb_D[\xd H]=\left\{(x^a,\zx_i,\dot{x}^b,\dot{\zx}_j):
\zr^i_a(x)\dot{x}^a=\frac{\pa H}{\pa\zx_i}(x,\zx)\,,\ c^k_{ab}(x)\zx_k\dot{x}^b -
\zr_a^i(x)\dot{\zx}_i=\frac{\pa H}{\pa x^a}(x,\zx)\right\}
$$
and
$$\ze_D[\xd L]=\left\{(x^a,\zx_i,\dot{x}^b,\dot{\zx}_j):\exists y\ \left[\zx_i=\frac{\pa L}{\pa y^i}(x,y)\,,\
\zr^i_a(x)\dot{x}^a=y^i\,,\ c^k_{ab}(x)\zx_k\dot{x}^b - \zr_a^i(x)\dot\zx_i=\frac{\pa L}{\pa
x^a}(x,y)\right]\right\}\,.
$$
The implicit Euler-Lagrange equations (Euler-Lagrange relations) take the form
\be\label{EL1}\qquad\zr^i_a(x)\frac{\xd x^a}{\xd t}(x)=y^i, \quad
\zr^i_a(x)\frac{\xd}{\xd t}\left(\frac{\partial L}{\partial y^i}\right)(x,y)= c^k_{ab}(x)\frac{\xd x^b}{\xd
t}(x) \frac{\partial L}{\partial y^k}(x,y) - \frac{\partial L}{\partial x^a}(x,y)\,.
\ee

Of course, for the canonical symplectic structure $\zw_M=\xd p_a\we\xd x^a$ on $E^\ast=\sT^\ast M$ we get the
classical dynamics as above. But also in the case of a regular presymplectic form of rank $r$,
$$\zw=\sum_{a\le r}\xd p_a\we\xd x^a\,,$$
we get the equations for the presymplectic reduction by the characteristic distribution to the reduced
symplectic form: the coordinates $x^a$ and $\dot{x}^a$ with $a>r$ are simply forgotten,
$$\frac{\xd}{\xd t}\left(\frac{\partial L}{\partial \dot{x}^a}\right)(x,\dot{x})=
- \frac{\partial L}{\partial x^a}(x,\dot{x})\,,\quad a\le r\,.$$
\end{example}

\begin{example}\label{e8} {\bf (Non-autonomous systems)}

\medskip\noindent
Consider the affine Dirac algebroid $D_0$ on $E_0=E\ti\R$ described  in example \ref{e5}, for the $\zP$-graph
Dirac algebroid $D=D_\zP$ on $E$, as in example \ref{e6}. In coordinates,
$$D_0=\{(x^0,x^a,\zx_i,\dot{x}^0,\dot{x}^b,\dot{\zx}_j,p_0,p_c,y^k):
\dot{x}^0=1\,,\ \dot{x}^b=\zr^b_k(x)y^k\,,\ \dot{\zx}_j=c^k_{ij}(x) y^i\zx_k - \zr^a_j(x) p_a\}\,.
$$
For a Lagrangian $L:E\ti\R\ra\R$ we get the Tulczyjew differential ${\zL_{D_0}^L}:E_0\rel\sT E_0^\ast$ of $L$
which is the map which in coordinates reads
$${\zL_{D_0}^L}(x^0,x^a,y^i)=(x^0,x^a,\zx_i,\dot{x}^0,\dot{x}^b,\dot{\zx}_j)$$
such that $$\zx_i=\frac{\pa L}{\pa y^i}(x^0,x^a,y^i)\,,\ \dot{x}^0=1\,,\ \dot{x}^b=\zr^b_k(x)y^k\,,\
\dot{\zx}_j=c^k_{lj}(x) y^l\frac{\pa L}{\pa y^k}(x^0,x^a,y^i) + \zr^b_j(x)\frac{\pa L}{\pa
x^b}(x^0,x^a,y^i)\,.
$$
Identifying $x^0$ with the time parameter $t$, we get the corresponding Euler-Lagrange equations in the form
$$\frac{\xd {x}^b}{\xd t} =\zr^b_k(x)y^k\,,\
\frac{\xd}{\xd t}\left(\frac{\pa L}{\pa y^j}(t,x^a,y^i)\right)=c^k_{lj}(x) y^l\frac{\pa L}{\pa y^k}(t,x^a,y^i)
+ \zr^b_j(x)\frac{\pa L}{\pa x^b}(t,x^a,y^i)\,.
$$
This is exactly the Euler-Lagrange equation on a skew algebroid for time-dependent Lagrangians. Such equations
have been obtained also as the Euler-Lagrange equations for affgebroids \cite{GGU2,GGU4,IMMS,IMPS,MMS1}. For
the canonical Lie algebroid $E=\sT M$, we get
$$\frac{\xd x^a}{\xd t}=\dot{x}^a, \quad
\frac{\xd}{\xd t}\left(\frac{\partial L}{\partial \dot{x}^a}(t,x,\dot{x})\right)=\frac{\partial L}{\partial
x^a}(t,x,\dot{x})\,.
$$
\end{example}

\begin{example} {\bf (Nonholonomic constraints)}

\medskip\noindent
Consider once more the Dirac algebroid $D_\zP$ associated with a linear bivector field $\zP$ on $E^\ast$, as
described in example \ref{e1}. Consider also a {\it nonholonomic constraint} defined by a vector subbundle $V$
of $E$ supported on a submanifold $S\subset M$. Using coordinates $(x^a)=(x^\za,x^A)$ in $M$, so that $S$ is
given locally by $x^A=0$, and linear coordinates $(y^i)$ in the fibers of $E$, so that $y=(y^i)=(y^\zi,y^I)$
and the subbundle $V$ is defined by the constraint $y^I=0$, on $\cT E^\ast$ we have then local coordinates
$(x^a,\dot{x}^b,\dot{\zx}_l,p_c,y^\zi,y^I)$, with decompositions $(\zx_k)=(\zx_\zk,\zx_K)$ and
$(\dot{\zx}_l)=(\dot{\zx}_\zl,\dot{\zx}_L)$ associated with the decomposition $(y^i)=(y^\zi,y^I)$. The Dirac
algebroid induced from $D_\zP$ by the nonholonomic constraint $V$ in local coordinates reads
$$D_\zP^V=\{(x^a,\zx_i,\dot{x}^b,\dot{\zx}_j,p_c,y^k):x^A=0\,,\ \dot{x}^b=\zr^b_\zi(x)y^\zi\,,\ \dot{\zx}_\zk=
c^j_{\zi \zk}(x) y^\zi\zx_j-\zr^a_\zk(x) p_a\,,\ y^I=0\}\,.
$$
The Tulczyjew differential ${\zL_{D_\zP^V}^L}$ associated with a Lagrangian $L:E\ra\R$ is defined on $V$ and
associates with every $(x^\za,0,y^\zi,0)\in V$ the set
\bea\nn{\zL_{D_\zP^V}^L}(x^\za,0,y^\zi,0)&=&\{(x^\za,0,\zx_i,\dot{x}^b,\dot{\zx}_j)\in\sT E^\ast:
\dot{x}^b=\zr^b_\zi(x)y^\zi\,,\ \zx_i=\frac{\pa L}{\pa y^i}(x^\za,0,y^\zi,0)\,,\\ \label{ph} &&\dot{\zx}_\zk=
c^j_{\zi \zk}(x) y^\zi\frac{\pa L}{\pa y^j}(x^\za,0,y^\zi,0)+\zr^a_\zk(x)\frac{\pa L}{\pa
x^a}(x^\za,0,y^\zi,0)\}\,.
\eea
Note that the coordinates $\dot{\zx}_A$ of points from ${\zL_{D_\zP^V}^L}(x^\za,0,y^\zi,0)$ are arbitrary. Curves
${\zL_{D_\zP^V}^L}$-related to a curve $\zg(t)=(x^\za(t),0,y^\zi(t),0)$ in $V$ have thus arbitrary coordinates
$\dot{\zx}_I$, but the remaining coordinates, if the curve is admissible, satisfy the {\it nonholonomically
constrained Euler-Lagrange equations}:
\bea\label{ELnh} &&x^A=0\,,\ y^I=0\,,\ \frac{\xd x^a}{\xd t}=\zr^a_\zi(x^\za,0)y^\zi\,, \\
&&\frac{\xd}{\xd t}\left(\frac{\partial L}{\partial y^\zk}(x^\za,0,y^\zi,0)\right)= c^j_{\zi\zk}(x^\za,0)
y^\zi\frac{\partial L}{\partial y^j}(x^\za,0,y^\zi,0) + \zr^a_\zk(x^\za,0)\frac{\partial L}{\partial
x^a}(x^\za,0,y^\zi,0)\,.\nn
\eea
Note that a minimal integrability requirement is the first integrability condition for $D_\zP^V$, saying that
$\zr^A_\zi(x^\za,0)=0$.

The constraint $V$ is {\it generalized holonomic} if, independently on the Lagrangian, the above equations
depend on the restriction of $L$ to $V$ only. Hence, $c^I_{\zi\zk}(x^\za,0)=0$ and $\zr^A_\zk(x^\za,0)=0$, so
that $V$ is generalized holonomic if and only if $V$ is a subalgebroid of $E$. In the classical situation of a
canonical Lie algebroid $\sT M$, the constraint $V$ is generalized holonomic if and only if $V$ is involutive,
for instance $V=\sT M_0$ for a submanifold $M_0$ in $M$. This is the traditional understanding of being
holonomic.
\end{example}

\begin{example} {\bf (Affine constraints)}

\medskip\noindent
We can perform a similar procedure with an affine nonholonomic constraint instead of the linear one. Let us
distinguish one variable $y^0$ from $y^I=(y^0,y^{\bar{I})}$ such that the affine constraint $A\subset E$ is
defined by $x^A=0\,,\ y^0=1\,, \ y^{\bar{I}}=0$. The model vector bundle $V=\sv(A)$ is as above and, as easily
checked, the constrained Euler-Lagrange equations are
\bea\nn &x^A=0\,,\ y^0=1\,,\ y^I=0\,,&\ \frac{\xd x^a}{\xd t}=\zr^a_0(x^\za,0)+\zr^a_\zi(x^\za,0)y^\zi\,, \\
\label{ELnha}&\frac{\xd}{\xd t}\left(\frac{\partial L}{\partial
y^\zk}(x^\za,0,y^\zi,0)\right)=&\left(c^j_{0\zk}(x) +c^j_{\zi\zk}(x^\za,0) y^\zi\right)\frac{\partial
L}{\partial y^j}(x^\za,0,y^\zi,0)+\\&& \zr^a_\zk(x^\za,0)\frac{\partial L}{\partial
x^a}(x^\za,0,y^\zi,0)\,,\nn
\eea
exactly as in \cite{GG}.
\end{example}

\begin{example} {\bf (Rolling disc)}

\medskip\noindent
To show how our method of Dirac algebroid works for an explicit constrained system, let us reconsider the case
of vertical rolling disc on a plane studied in \cite{YM1}. The position configuration space is the Lie group
$N=\R^2\ti S^1\ti S^1$ with coordinates $(x^1,x^2,\zvy,\zf)$. The Lagrangian on $\sT N$ in the adapted
coordinates $(x^1,x^2,\zvy,\zf,\dot{x}^1,\dot{x}^2,\dot{\zf},\dot{\zvy})$ reads
\be\label{lag}L(x^1,x^2,\zf,\zvy,\dot{x}^1,\dot{x}^2,\dot{\zf},\dot{\zvy})=\frac{1}{2}m\left((\dot{x}^1)^2+(\dot{x}^2)^2\right)+
\frac{1}{2}J_1\dot{\zf}^2+\frac{1}{2}J_2\dot{\zvy}^2\,.
\ee
The kinematic constraint due to the rolling contact without slipping on the plane is
$$\dot{x}^1=R\,\dot{\zvy}\cos{\zf}\,,\ \dot{x}^2=R\,\dot{\zvy}\sin{\zf}\,.
$$
Since the Lagrangian and the constraints are invariant with respect to translation with respect to
$x^1,x^2,\zvy$, we have an obvious Lie algebroid reduction to $E=\sT N/(\R^2\ti S^1)=\sT\R\ti \R^3$ which is a
vector bundle of rank 4 over $S^1$ with coordinates
$$(\zf,\dot{\zf},\dot{x}^1,\dot{x}^2,\dot{\zvy})\,,
$$
associated with the basis of (global) sections $(f_1=\pa_\zf,f_2,f_3,f_4)$, where $f_2,f_3,f_4$ come from the
reductions of $\pa_{{\zvy}}$, $\pa_{x^1}$, $\pa_{x^2}$, respectively. The anchor $\zr:E\ra\sT S^1$ is the
projection onto $\sT S^1$, and all the basic sections commute. The reduced Lagrangian we will denote also $L$,
as it takes values exactly like in (\ref{lag}).

The constraint subbundle $V$ of $E$ is spanned by the sections $e_1=f_1$ and $e_2=f_2+R\cos{\zf}\cdot
f_3+R\sin{\zf}\cdot f_4$, so we can use the basis $e_1=f_1,e_2,e_3=f_3,e_4=f_4$, and the corresponding
coordinates $(\zf,y)$ on $E$. The basis $(e_1, e_2, e_3, e_4)$ induces the coordinate system $(\zf, \zx)$ in
$E^\ast$ and adapted coordinates $(\zf,\zx,\dot\zf,\dot\zx)$ in $\sT E^\ast$ and $(\zf,\zx,p,y)$ in $\sT^\ast
E^\ast$. The constraint is now described by the equations $y^3=y^4=0$, but we get non-trivial commutation
relations
$$[e_1,e_2]=R\cos{\zf}\cdot e_4-R\sin{\zf}\cdot e_3.
$$
In other words, the corresponding Poisson tensor $\zP$ on $E^\ast$ in the adapted coordinates reads
$$\zP=R(\cos{\zf}\cdot\zx_4-\sin{\zf}\cdot\zx_3)\pa_{\zx_1}\we\pa_{\zx_2}+\pa_{\zx_1}\we\pa_{\zf}\, ,
$$
and the Dirac structure induced by the constraints is
\begin{multline}
D_\zP^V=\{\, (\zf, \zx, \dot\zf, \dot\zx, p, y):\quad
y^3=y^4=0,\; \dot\zf=y^1,  \\
\dot\zx_1=Ry^2\zx_3\sin\zf-R y^2\zx_4\cos\zf-p,\quad \dot\zx_2=-R y^1\zx_3\sin\zf+R y^1\zx_4\cos\zf\,\}.
\end{multline}
Hence, the nonholonomically constrained Euler-Lagrange equations (\ref{ELnh}) take the form
\bea\label{hhh} &&y^3=y^4=0\,,\ \frac{\xd\zf}{\xd t}=y^1\,,\\
\nn&&\frac{\xd}{\xd t}\left(\frac{\pa L}{\pa y^1}(\zf,y^1,y^2,0,0)\right)=\\
\nn &&R\sin{\zf}\cdot y^2\frac{\pa L}{\pa y^3}(\zf,y^1,y^2,0,0)-R\cos{\zf}\cdot y^2\frac{\pa L}{\pa y^4}
(\zf,y^1,y^2,0,0)+\frac{\pa L}{\pa \zf}(\zf,y^1,y^2,0,0)\,,\\
\nn &&\frac{\xd}{\xd t}\left(\frac{\pa L}{\pa y^2}(\zf,y^1,y^2,0,0)\right)=-R\sin{\zf}\cdot y^1\frac{\pa
L}{\pa y^3}(\zf,y^1,y^2,0,0)+R\cos{\zf}\cdot y^1\frac{\pa L}{\pa y^4}(\zf,y^1,y^2,0,0)\,.
\eea
Since
$$\dot{x}^1=y^3+Ry^2\cos{\zf}\,,\ \dot{x}^2=y^4+Ry^2\sin{\zf}\,,\ \dot{\zf}=y^1\,,\ \dot{\zvy}=y^2\,,$$
the Lagrangian in coordinates $(\zf,y)$ reads
\beas &&L(\zf,y^1,y^2,y^3,y^4)=\frac{1}{2}m\left((y^3)^2+(y^4)^2\right)+\\
&&\frac{1}{2}J_1(y^1)^2+\frac{1}{2}(mR^2+J_2)(y^2)^2+mRy^2(y^3\cos{\zf}+y^4\sin{\zf})\,.
\eeas
and, as show the straightforward calculations, the Euler-Lagrange equations (\ref{hhh}) take the form
\be\label{hhh1} y^3=y^4=0\,,\ \frac{\xd\zf}{\xd t}=y^1\,,\ (mR^2+J_2)\frac{\xd y^2}{\xd t}=0\,,\
J_1\frac{\xd y^1}{\xd t}=0\,.
\ee
Going back to the original coordinates, we get finally
\be\label{gg}\dot{x}^1=R\,\dot{\zvy}\cos{\zf}\,,\ \dot{x}^2=R\,\dot{\zvy}\sin{\zf}\,,\ \ddot{\zf}=0\,,\ \ddot{\zvy}=0\,,
\ee
with obvious explicit solutions.

If the phase dynamics is concerned, in view of (\ref{ph}), we get that $\ze_{D_\zP^V}[\xd
L]={\zL_{D_\zP^V}^L}(V)$ is parametrized by $(\phi,y^1,y^2)$ as follows:
\beas \ze_{D_\zP^V}[\xd L]&=&\{(\zf,\zx_1,\zx_2,\zx_3,\zx_4,\dot{\zf},\dot{\zx}_1,\dot{\zx}_2,\dot{\zx}_3,\dot{\zx}_4):
\zx_1=J_1y^1\,,\ \zx_2=(mR^2+J_2)y^2\,,\\
&& \zx_3=mRy^2\cos{\zf}\,,\ \zx_4=mRy^2\sin{\zf}\,,\ \dot{\zf}=y^1\,,\ \dot{\zx}_1=0\,,\ \dot{\zx}_2=0\}\,.
\eeas
Here $\dot{\zx}_3$ and $\dot{\zx}_4$ are arbitrary, but the integrability condition allows us to describe them
as well. Let us note that the phase space $\Ph_D^L\subset E^\ast$ is defined by the equations
\be\label{e61}\zx_3=\zm\cos{\zf}\cdot\zx_2,\quad \zx_4=\zm\sin{\zf}\cdot\zx_2,\,
\ee
where $\zm=\frac{mR}{mR^2+J_2}$. Hence, the first integrability condition gives
$$\dot{\zx}_3=-\zm\zx_2\sin{\zf}\cdot\dot{\zf}=-\frac{\zm}{J_1}\zx_1\zx_2\sin{\zf}\,,\quad
\dot{\zx}_4=\zm\zx_2\cos{\zf}\cdot\dot{\zf}=\frac{\zm}{J_1}\zx_1\zx_2\cos{\zf}\,.
$$

The dynamics is Hamiltonian, since the Hamiltonian
$$H(\zf,\zx)=\frac{1}{2J_1}(\zx_1)^2+\frac{1}{2J_2}(\zx_2-R\zx_3\cos\zf-R\zx_4\sin\zf)^2+\frac{1}{2m}((\zx_3)^2+(\zx_4)^2)$$
defined on $E^\ast$ induces the dynamics $\zb_{D_\zP^V}[\xd H]=\ze_{D_\zP^V}[\xd L]$. The equality can be
checked by straightforward calculations. Let us only note that, since
\beas && y^3=\frac{\partial H}{\partial \zx_3}=-\frac{R}{J_2}\zx_2\sin\zf+\left(\frac{R^2}{J_2}\cos^2\zf+\frac{1}{m}\right)\zx_3+
\frac{R^2}{J_2}\zx_4\sin\zf\cos\zf\,, \\
&& y^4=\frac{\partial H}{\partial \zx_4}=-\frac{R}{J_2}\zx_2\sin\zf+\frac{R^2}{J_2}\zx_3\cos\zf\sin\zf+
\left(\frac{R^2}{J_2}\sin^2\zf+\frac{1}{m}\right)\zx_4\,,\eeas imposing the conditions $y^3=0$, $y^4=0$ of the
Dirac structure, we recover the Hamiltonian constraints (\ref{e61}).
\end{example}

\section{Concluding remarks}
We have introduced the concepts of Dirac  and Dirac-Lie algebroid as a natural common generalization of a skew
(resp., Lie) algebroid and a linear presymplectic structure. Aside from its interesting geometrical structure,
Dirac algebroids, as well as their affine counterparts -- affine Dirac algebroids, provide a powerful
geometrical tool for description of mechanical systems by means of generalized Lagrangian and Hamiltonian
formalisms.

The kinematic configurations (quasi-velocities) form in this framework a subset of a vector bundle $\zt:E\ra
M$ and are related to the actual velocities from $\sT M$ by the so called anchor relation, while the phase
space is a subset of the dual bundle, $E^\ast$. The phase dynamics induced by a Lagrangian or a Hamiltonian is
an implicit dynamics in the phase space described by a subset of the tangent bundle $\sT E^\ast$, and the
associated Euler-Lagrange equations are defined by an implicit dynamics in $E$.

We proposed a well-described procedure of inducing a new Dirac algebroid out of a given one by imposing
certain linear constraints in the anchor relation (velocity bundle), that on the Lagrangian formalism level
corresponds to imposing nonholonomic constraints. Since imposing constraints we end up in a Dirac algebroid
again, our approach does not really distinguish between constrained and unconstrained systems, as well as
between regular and singular Lagrangians. Since the use of algebroids already includes reductions to the picture, our approach covers all main examples of mechanical systems: regular or singular, constrained or not, autonomous or non-autonomous etc.

The Dirac algebroid, especially the Dirac-Lie algebroids, possess a rich and intriguing geometrical structure whose
investigations have been started in the present paper. We are strongly convinced that these objects, as well
as their possible generalization, will allow us to find a proper intrinsic framework also for field theories and
other areas of mathematical physics.


\bigskip
\noindent Katarzyna Grabowska\\
Faculty of Physics,
                University of Warsaw \\
                Ho\.za 69, 00-681 Warszawa, Poland \\
                 {\tt konieczn@fuw.edu.pl} \\\\
\noindent Janusz Grabowski\\Institute of Mathematics, Polish Academy of Sciences\\\'Sniadeckich 8, P.O. Box
21, 00-956 Warszawa,
Poland\\{\tt jagrab@impan.pl}\\\\

\end{document}